\documentclass[journal]{IEEEtran}
\usepackage{hyperref}
\usepackage{cite}
\usepackage[inline]{enumitem}
\usepackage{pifont}
\usepackage{mathmacros}
\usepackage{float}
\usepackage{graphicx}
\usepackage{xcolor}
\usepackage[capitalize]{cleveref}
\usepackage{subfig}
\usepackage[normalem]{ulem}
\usepackage{standalone}
\usepackage[font=small]{caption}
\usepackage{tikz}

\usepackage{pgfplots} 
\pgfplotsset{compat=1.16}

\usepgfplotslibrary{fillbetween,groupplots}

\theoremstyle{theorem}
\newtheorem{definition}{Definition}
\theoremstyle{theorem}

\theoremstyle{theorem}
\newtheorem{remark}{Remark}

\newcommand{\dham}{d^{(\mathrm{h})}}
\newcommand{\ds}{d^{(\mathrm{s})}}
\newcommand{\dc}{d^{(\mathrm{c})}}
\newcommand{\Span}[1]{\langle#1\rangle}
\newcommand{\Gr}{\mathrm{Gr}}

\begin{document}
%

\title{
Grassmannian-Coded Beamforming for mmWave Channel Sensing with Unknown Complex Path Gain
}


%
\author{Parthasarathi~Khirwadkar,
        Robin~Rajam\"{a}ki,
        and~Piya~Pal
\thanks{P. Khirwadkar and P. Pal are with the Department of Electrical and Computer Engineering, University of California San Diego, CA, 92093 USA}
\thanks{R. Rajam\"{a}ki is with Tampere University, Finland.}
\thanks{Manuscript received XX XX, 2026; revised XX XX, 2026.}}

\markboth{Journal of \LaTeX\ Class Files,~Vol.~13, No.~9, September~2025}%
{Shell \MakeLowercase{\textit{et al.}}: Bare Demo of IEEEtran.cls for Journals}
%



\maketitle

\begin{abstract}
This paper introduces a subspace-coding perspective to millimeter-wave channel sensing with a single RF chain when the complex channel gain is unknown. 
We show that in this case, candidate directions-of-arrival (DoAs) map naturally to subspaces through their beamspace responses, revealing an intrinsic Grassmannian geometry. This motivates \emph{beamspace Grassmannian codes} (BGCs), designed to reduce DoA error by maximizing the minimum subspace distance of the joint beamformer-array response. We identify two regimes: one in which existing Grassmannian packings are exactly realizable as BGCs when the angular grid matches the array size, and another in which realizability for finer grids is constrained by the array geometry. Our analysis establishes the joint roles of subspace distance and beamforming gain in sensing performance and motivates two complementary beamformer designs. Without prior DoA information, we develop spatially isotropic beamformers based on algebraic Grassmannian packings and modulation-based channel codes. With a known DoA region of interest, we design convolutional beamspaces that combine directional gain with favorable subspace distance. Numerical results demonstrate robust BGC performance for both on-grid and off-grid DoAs, supporting the effectiveness of the proposed Grassmannian framework for mmWave channel sensing.

\end{abstract}


\IEEEpeerreviewmaketitle

\section{Introduction}
\label{sec:Introduction}
Sensing is poised to become a defining capability of 6G, driven both by integrated sensing and communications (ISAC) \cite{zhang2021anoverview,gonzalezprelcic2024theintegrated,liu2022integrated} and the use of higher frequency bands, including {millimeter-wave} (mmWave) spectrum \cite{rappaport2015millimeter}. Such multiple-input multiple-output (MIMO) systems must meet stricter spectral, energy, and cost demands \cite{andrews20246g} while exploiting larger {antenna arrays}, raising challenges such as channel state acquisition and pilot overhead \cite{tataria20216g}. A fundamental challenge is channel sensing (estimating angular parameters) under hardware constraints, especially with hybrid/analog beamforming with fewer power-hungry radio-frequency (RF) chains \cite{Alkhateeb2014,noh2017multi,zhang2017codebook,Caire2019}. These induce a \emph{beamspace} model, where the receiver observes linear combinations of antenna signals, using a single RF chain. Hence, a key challenge is in designing {accurate}, low-overhead channel estimation for mmWave communication.

Several works address channel sensing {and beam alignment} under hardware constraints. Non-adaptive compressed sensing exploits mmWave channel sparsity to avoid exhaustive beam sweeps \cite{Indyk2018,Myers2019falp,marzi2016compressive,alkhateeb2015compressed}. Antenna selection has also been applied to sensing and ISAC \cite{wang2019dual,rajamaki2024sparse}, {with} \emph{sparse arrays}, improving resolution and identifiability \cite{wang2017coarrays,sarangi2023superresolution,amin2024sparsearrays}. Adaptive methods \cite{zhang2017codebook,noh2017multi,Javidi2019ie8n,Yu2022tddMassiveMIMO,Erkip2021} {use} feedback from past measurements to refine future probing. {Deep learning based methods likewise learn sensing policies and site-specific codebooks \cite{Andrews2022} from data, both adaptively \cite{sohrabi2022active,ParthICASSP} and non-adaptively \cite{andrews2021channel,Andrews2022,Yu2022tddMassiveMIMO}.}


Error-correcting codes (ECCs) have recently emerged as a promising tool for {low overhead} non-adaptive channel sensing {and beam alignment} \cite{Ekici2018,Love2019ssm,Dai2025,Yu2025ongrid}. The key idea is to embed algebraic coding structure directly into the sensing beamformers, thereby translating the error-correcting capability of the code into robustness of direction-of-arrival (DoA) estimation. In \cite{Love2019ssm,Dai2025,Yu2025ongrid}, each angle on a discrete grid is associated with a unique codeword, and {the beamformers are constructed such that the beamspace measurements realize the codeword corresponding to the true DoA.} 
For example, \cite{Yu2025ongrid} designs the beamformers 
{to realize binary phase shift keying} (BPSK)-modulated Reed--Muller codes { as its beamspace response}. Relative to compressed sensing approaches, this coding-theoretic viewpoint offers two important advantages: deterministic beamformer construction, avoiding random beam patterns \cite{alkhateeb2015compressed,marzi2016compressive}, and performance guarantees that follow directly from structural properties of the underlying code.



A separate line of signal processing work exploits convolutional beamspace (CBS) structure to combine advantages of antenna- and beamspace processing for channel sensing. {This is particularly useful when probing selected spatial directions, e.g., when prior DoA information is available.} Introduced in \cite{Vaidyanathan2020thf}, CBS constructs beamspace transformations using linear shifts of a base beamforming filter while preserving Vandermonde structure, thereby enabling classical subspace DoA estimators to operate directly on beamspace measurements. Subsequent work \cite{Vaidyanathan2024} studied CBS for mmWave sensing under uniform and non-uniform decimation, while \cite{Myers2025insector,Verhaegen2022} developed two-stage CBS strategies to improve angular resolution.

\subsection{Grassmannian Coding Perspective and Open Questions}
Despite recent advances in ECC- and CBS-based channel sensing with reduced RF chains, fundamental questions about  codebook design and their performance remain open. 
{In particular, \cite{Yu2025ongrid} showed that when the dominant-path channel coefficient is known, the decoding error probability of ECC-based sensing beamformers directly determines the DoA estimation error probability.} 
Probability of error is, in general, a well-established reliability criterion in DoA estimation, particularly when recovery entails selecting among angular hypotheses. Related formulations include support-recovery error probability, tolerance-based DoA error probability, and outlier probability \cite{tang2010performance,Javidi2019ie8n,Athley2005}.
Under such partial {channel state information (CSI)}, ECC performance is naturally governed by the Hamming distance between codewords. But is Hamming distance still the right design metric when the channel coefficient is unknown? As we will show, an unknown complex path gain causes each candidate DoA to correspond to a one-dimensional subspace (invariant to scaling due to the unknown gain), revealing an intrinsic \emph{Grassmannian geometry} \cite{strohmer2003grassmannian,conway1996packing}. Consequently, \emph{subspace distance}, rather than Hamming distance, becomes the natural codebook metric, and beamspace channel sensing reduces to a structured \emph{Grassmannian line-packing} problem \cite{strohmer2003grassmannian,conway1996packing,Giannakis2005,love2003grassmanian,love2005limited,love2006limited,mukkavilli2003beamforming}, but with additional structure induced by the antenna array manifold.

This paper develops a unified beamspace-design framework based on subspace codes, channel codes, and convolutional beamspaces through the lens of {\emph{beamspace Grassmannian codes} (BGCs)}. 
The framework connects to the extensive literature on Grassmannian packings \cite{conway1996packing,strohmer2003grassmannian}, including noncoherent communications \cite{Marzetta2000,zheng2002communication}, limited-feedback beamforming \cite{mukkavilli2003beamforming,love2003grassmanian,love2008anoverview,khoshnevis2012bit}, and subspace coding \cite{KK,Mahdavifar2022}. In these settings, subspaces serve as codewords for noncoherent transmission, low-overhead feedback, or robustness to channel perturbations. The distinguishing feature of mmWave sensing is its highly structured channel geometry, where the array manifold and beamformer jointly determine the induced Grassmannian code.

Our recent work \cite{mahdavifar2024subspace}, inspired by \emph{analog subspace codes} \cite{Mahdavifar2022}, exploited this structure in antenna space and showed that each DoA is naturally encoded by the subspace spanned by the corresponding antenna{-space} measurements. A key result of \cite{mahdavifar2024subspace} is that estimation performance is governed by the \emph{minimum subspace distance} between DoA codewords, which is controlled by the array geometry; suitably designed sparse arrays can therefore outperform conventional uniform arrays. 
Extending this viewpoint to beamspace sensing with fewer RF chains {than antennas} introduces additional structural constraints. {As we will see,} beamspace codewords are jointly induced by the beamforming matrix and array manifold, leading to a nontrivial \emph{realizability} condition, where existing Grassmannian line packings are useful for mmWave sensing when they are compatible with the employed array geometry.
\subsection{Contributions and Organization} \label{sec:contributions}
Our main contributions are summarized as follows:

\begin{itemize}

\item \textbf{Role of subspace distance in channel sensing:} We show that channel sensing with unknown complex {path} gain $\alpha$ is naturally a noncoherent subspace decoding problem, where each candidate DoA is represented by the one-dimensional subspace spanned by its beamspace response. We establish that the minimum-distance (MD) subspace decoder is equivalent to the widely studied Maximum-Likelihood (ML) DoA estimator \cite{stoica1989music,van2002optimum} 
when $\alpha$ and the DoA are deterministic unknowns. For scenarios without any prior DoA information, we introduce \emph{spatially isotropic} beamformers with uniform total beamgain across directions, and explicitly characterize how beamforming gain and minimum subspace distance govern the error probability.

\item \textbf{Realizability of BGCs:} Motivated by this geometry, we characterize which Grassmannian line packings can be realized as BGCs. 
{Two regimes unique to our problem emerge:
}
(a) the \emph{fully realizable} regime, 
where arbitrary packings can be implemented, and (b) the \emph{partially realizable} regime, 
where exact realization requires an array-dependent row-space condition. 
We show that {constructions based on so-called Singer and Bose--Chowla sets} remain realizable in both regimes through low-complexity antenna selection.
Numerically, these designs remain effective for both off-grid DoAs and 
angular grids refined beyond the number of antennas.

\item \textbf{Subspace distance of modulated channel codes:} We 
connect conventional channel code-based designs to 
the 
Grassmannian viewpoint to evaluate their efficacy when $\alpha$ is unknown. For {BPSK-mapped} binary codebooks, 
 we derive the exact minimum subspace distance in terms of both the minimum and maximum Hamming distances. 
{We show that} large minimum Hamming distance alone is insufficient for sensing with unknown $\alpha$, since complementary binary codewords represent the same line. This result provides a principled basis for selecting and pruning channel codes, including Reed--Muller codes, for subspace-based channel sensing.

\item \textbf{Spatially selective BGCs:} Finally, we extend BGCs using convolutional beamspaces to settings with prior DoA information. Here, the shift pattern controls minimum subspace distance, while the filter controls beamforming gain. Singer and Bose--Chowla shift patterns can therefore preserve favorable distance properties while concentrating sensing gain over a prescribed angular region. 
\end{itemize}

The paper is organized as follows. \cref{sec:background} introduces the problem formulation and reviews Grassmannian line packing. \cref{sec:relation_ML_MAP_minimumdist} develops the BGC framework and establishes the connection between ML and minimum-subspace-distance DoA estimation. \cref{sec:design_isotropic_BGC} studies spatially isotropic BGCs constructed from Grassmannian packings and channel codes under the two realizability conditions, without prior DoA information, while \cref{sec:harnessing_conv_beamspace} develops spatially selective BGCs when such prior DoA information is available. {\cref{sec:numerical} numerically validates the theoretical results and examines extensions to off-grid DoAs.} \cref{sec:conclusions} concludes the paper.

\section{
Problem Formulation and Background Review}\label{sec:background}

\subsection{System model} \label{sec:system_model}
We consider a mmWave 
communication system between a base station (BS) equipped with a single RF chain connected to a linear array of $N_a$ antennas serving $U$ single-antenna user equipment (UE) \cite{Javidi2019ie8n,Yu2025ongrid}. The antenna location at BS are denoted by $\SS\triangleq\{d_1,d_2,\dots,d_{N_a}\}$, 
{where the $d_i$'s are integer multiples of half the carrier wavelength $\lambda/2$.} 
The uplink channel between BS and a given UE consisting of a single line-of-sight path is of the form 
\cite{Sayeed2016,Alkhateeb2014}
\begin{align} \label{eq:channel_model}
    \h = \alpha\a_{\SS}(\sin(\theta)),
\end{align}
where $\theta$ is the DoA, $\a_{\SS}(f) = [e^{j\pi d_1 f},e^{j\pi d_2 f},\dots,e^{j\pi d_{N_a}f}]$ is the array steering vector corresponding to spatial frequency $f=\sin(\theta)$, and $\alpha$ is the \emph{unknown} channel coefficient which combines transmit power and propagation loss. During the uplink pilot training phase, each UE $u'\!=\!1,\dots,U$ transmits an 
orthonormal spreading sequence $\s_{u'}(t)$ satisfying $\s_{u'}(t)^H\s_{u}(t)\!=\!\delta(u-u')$. To sense the channel corresponding to user $u$, the BS acquires scalar measurement $y_t\!\in\!\CC, t\!=\!1,\dots,T$ using beamforming vector $\w_t\!\in\!\CC^{N_a}$ and matched filtering with $\s_u(t)$ 
\cite{Javidi2019ie8n}: 
\begin{align} \label{eq:measurement_model} 
    y_t &= \w_t^H \left( \sum_{u'=1}^{U}\h_{u'}\s_{u'}(t)^T\right)\s_{u}(t)^* + \w_t^H\N_t\s_{u}(t)^*\nonumber\\
    &= \w_t^H\h_u+ \w_t^H\n_t
    = \w_t^H\a_{\SS}(f)\alpha + \w_t^H\n_t,
\end{align}
where $f=\sin(\theta)$ is the DoA of user $u$, and $\n_t \triangleq \N_t\s_{u'}(t)^*$ is the noise after matched filtering. 
We express all received measurements $\y = [y_1,\dots,y_T]^T$ jointly as
\begin{align} \label{eq:measurement_model_joint}
    \y = \alpha\W\a_{\SS}(f) + \z,
\end{align}
where $\W = [\w_1,\dots,\w_T]^H$ and $\z= [\w_1^H\n_1,\dots,\w_T^H\n_T]^T$.

\textbf{Problem formulation:}  Our objective is to estimate DoA $f$ from  $\y$ in \eqref{eq:measurement_model_joint} \emph{without prior knowledge of $\alpha$}. To this end, we design the beamforming matrix $\W$ under the following two scenarios: 
\begin{enumerate*}[label=(\roman*)]
    \item no prior knowledge on $f$, and
    \item partial knowledge where $f \in \mathcal{F} \subset [-1, 1)$, enabling region-focused beamforming.
\end{enumerate*} 
These cases are addressed in \cref{sec:design_isotropic_BGC,sec:harnessing_conv_beamspace}, respectively. 
We model both $f$ and $\alpha$ as deterministic unknowns and do not impose statistical priors on either.
We make the following assumptions throughout the paper:
\begin{enumerate}[label=\textbf{A\arabic*}]
    \item\!\!\!: Direction-of-arrival $f$ lies on a uniform grid 
    $\cG = \{f_n = -1 + \frac{2(n-1)}{N_g},n=1,\dots,N_g\}$ with $N_g$ grid points.\label{a:grid}
    \item\!\!\!: Antennas at the base station are arranged in a uniform linear array (ULA) such that $\SS=\UU \triangleq \{0,1\dots,N_a-1\}$.\label{a:ula}
    \item\!\!\!: Entries of 
    $\N_t$ are i.i.d 
    complex Gaussian  $\CN(0,\sigma^2), \forall t$.\label{a:noise}
\end{enumerate}
We also make the following beamformer design choices:
\begin{enumerate}[label=\textbf{D\arabic*}]
    \item\!\!\!:  Beamforming matrix $\W$ satisfies $\W\a_{\SS}(f_n)\!\neq\!\zero,\ \forall f_n\!\in\!\cG$.\label{d:identifiability}
    \item\!\!\!: Beamformers satisfy $\|\w_t\|_2 = 1$ for all $t=1,\dots,T$.\label{d:unitnorm}
\end{enumerate}
Assumption \labelcref{a:grid} is standard in the literature \cite{Alkhateeb2014,Yu2025ongrid} and adopted here for simplicity, with relaxations examined numerically in \cref{sec:numerical}. 
We focus on the regime \(N_g\geq N_a\), which includes both the critically sampled case \(N_g=N_a\) and the oversampled regime \(N_g>N_a\) relevant to high-resolution DoA estimation. This automatically ensures that angle estimation is performed on a grid with sufficient resolution relative to the number of antennas and/or aperture.
\labelcref{d:identifiability} ensures that every grid point is identifiable, while \labelcref{d:unitnorm} together with assumption \labelcref{a:noise} imply that noise after matched filtering and beamforming in \eqref{eq:measurement_model_joint} is i.i.d zero-mean uncorrelated complex Gaussian i.e. $\z\sim \CN(0,\sigma^2\I)$. 
In the next subsection, we review the concept of Grassmanian line packing which will play an important role in our discussion on designing $\W$ for robust channel sensing.



\subsection{{Comparative Review of Grassmanian Line Packing}} \label{sec:review_grassmanian}

Grassmannian line packings broadly refer to an optimal way of arranging a given number (say $N$) of one-dimensional subspaces (or lines) in an $M$-dimensional complex linear vector space $\CC^{M}$. There is a vast body of literature on Grassmannian line packings and their mathematical properties \cite{conway1996packing,strohmer2003grassmannian}, which have found extensive use in wireless communication and coding \cite{Marzetta2000,zheng2002communication,KK,Mahdavifar2022}, in particular 
in 
limited feedback beamforming \cite{love2003grassmanian,mukkavilli2003beamforming,love2008anoverview}. To elucidate the novel and unique aspects of our problem and results, we review the key mathematical properties of Grassmannian line packings, and how they have been used in limited feedback beamforming.

\subsubsection{Key Mathematical Properties of Grassmannian Packings} 

The complex Grassmannian manifold $\Gr_{\CC}(d,M)$ is defined as
\begin{align*}
   \Gr_{\CC}(d,M) \triangleq \left\{ \mathcal U\subset \CC^M: \dim(\mathcal U)=d \right\},
\end{align*}
which is the set of all \(d\)-dimensional subspaces of \(\CC^M\). A finite subset \(\mathcal C\subset\operatorname{Gr}_{\CC}(d,M)\) is referred to as a Grassmannian codebook \cite{strohmer2003grassmannian,Mahdavifar2022}. In this paper, we will focus primarily on Grassmanian codebooks of dimension $d=1$, i.e., codebooks of lines passing through the origin. 
An element of $\mathrm{Gr}_{\CC}(1,M)$ can be represented by the span of a non-zero vector $\u$ denoted as $\Span\u 
$. This structure captures the equivalence 
among vectors in $\CC^{M}$ under global scaling $r$ and phase ambiguity $\phi$ as:
\begin{align*}
    \u \sim re^{j\phi} \u,\ \forall r>0, \phi\in [0,2\pi).
\end{align*}  
In other words, two unit vectors $\u,\v$ 
{differing by a global phase represent the same Grassmannian line.} 
The chordal distance between lines \(\Span\u,\Span\v\in \Gr_{\CC}(1,M)\) is defined as $\dc(\Span\u, \Span\v)\!\triangleq\!\sqrt{1\!-\!\frac{|\u^H\v|^2} {\|\u\|_2^2\|\v\|_2^2}} $. The squared chordal distance, also referred to as  ``subspace distance", is defined as  \cite{Mahdavifar2022,mahdavifar2024subspace}
\begin{align} 
\ds(\Span\u, \Span\v) &\triangleq \left(\dc(\Span\u, \Span\v)\right)^2 = 1- \frac{|\u^H\v|^2} {\|\u\|_2^2\|\v\|_2^2}. \label{eq:subspace_distance_def} 
\end{align} 
Although \(\dc\) is a distance metric, its square $\ds$ is a quasi-distance since it satisfies the generalized triangle inequality \cite{Mahdavifar2022}. The \emph{minimum subspace distance} of a Grassmannian line codebook $\cC\subset \mathrm{Gr}_{\CC}(1,M)$ of cardinality \(|\cC|\!=\!N\) is defined as
\begin{align} 
\ds_{\min}(\mathcal C)\triangleq \min_{{i,j\in[N],i\neq j}} \ds(\Span{\u_i},\Span{\u_j}). \label{eq:min_squared_chordal_distance} 
\end{align}
Hence, for a fixed ambient dimension \(M\) and codebook cardinality \(N\), the complex Grassmannian line-packing problem is to obtain a codebook $\cC^\star \subset \Gr_{\CC}(1,M) $ such that 
\begin{align} 
\cC^\star = \argmax_{{ \mathcal C\subset\mathrm{Gr}_{\CC}(1,M), |\mathcal C|=N }} \ds_{\min}(\mathcal C) \label{eq:grassmannian_line_packing}.
\end{align}

 The Grassmannian geometry, 
 is 
 invariant under unitary rotations. In particular, let \(\mathbf U\in\CC^{M\times M}\) be unitary, then for any \( \Span{\u},\Span{\v} \in\mathrm{Gr}_{\CC}(1,M)\), we have
 \begin{align*}
  \ds(\Span{\U\u},\Span{\U\v}) = \ds(\Span{\u},\Span{\v}).   
 \end{align*}
 Consequently, the rotated codebook
 \( \mathbf U\cC \triangleq \{\Span{\U\u_i}: i=1,\dots,N\} \)
 has the same minimum distance as \(\cC\), i.e., $\ds_{\min}(\mathbf U\cC) = \ds_{\min}(\cC).$ Thus, codebooks that differ only by a common unitary rotation are geometrically equivalent Grassmannian line packings. This indicates that it is possible to generate multiple Grassmanian codebooks with the same minimum distance property through unitary transformation. 

An upper bound on minimum subspace distance of any Grassmannian line codebook $\cC \subset\Gr_{\CC}(1,M)$ with $|\cC| = N$ is given by the Welch bound \cite{welch1967lower}, which can be expressed as
\begin{align}\label{eq:welch}
    \ds_{\min}(\cC) = 1 - \max_{i,j\in [N], i\neq j} |\u_i^H\u_j|^2 &\leq 1 - \frac{N-M}{M(N-1)}.
\end{align}
The Welch bound is attainable by $\cC\in \Gr_{\CC}(1,M)$ only when $N\leq M^2$ \cite{strohmer2003grassmannian}. Hence, in the quadratic regime $N=O(M^2)$, it 
provides a useful benchmark to evaluate the performance of Grassmannian line packings, since a line packing that attains the Welch bound is clearly an optimal packing. 
Welch-bound-achieving packings may not exist though for every $M, N$ \cite{strohmer2003grassmannian}.

\subsubsection{Constructions}
A vast body of work addresses optimal Grassmannian line packings \cite{conway1996packing,strohmer2003grassmannian,Giannakis2005,dhillon2008constructing-e8f,medra2014flexible}. Explicit constructions exploit equiangular tight frames, conference matrices, harmonic frames, and combinatorial difference sets \cite{strohmer2003grassmannian,Giannakis2005}. In particular, cyclic difference sets yield constant-modulus packings attaining the Welch bound for selected $(M,N)$ \cite{Giannakis2005,strohmer2003grassmannian}. Since Welch-bound achieving constructions are restricted to $N\leq M^2$ and are unavailable for general $(M,N)$, numerical alternatives include generalized Lloyd methods and alternating projections 
\cite{conway1996packing,Giannakis2005,dhillon2008constructing-e8f,medra2014flexible}.

A broadly applicable alternative to generate a Grassmanian line codebook on $\Gr_{\CC}(1,M)$  is random vector quantization (RVQ)  \cite{Love2007,raghavan2013ensemble}. {In limited feedback beamforming, given $M$ transmit} antennas and $B$ feedback bits, an RVQ codebook $\cC_{\rm RVQ}\subset \Gr_{\CC}(1,M_t)$ consists of codewords
\begin{align}
    \cC_{\mathrm{RVQ}} =\left\{ \Span{\f_1},\ldots,\Span{\f_{2^B}} \right\}, \label{eq:rvq_definition}
\end{align}
where the unit-norm $\f_i\in\CC^M$ are drawn isotropically and independently \cite{Love2007,raghavan2013ensemble}. {Such random Grassmannian line codes achieve nonvanishing minimum subspace distance with probability approaching one exponentially in dimension \cite{Mahdavifar2022}.} RVQ performance for limited-feedback beamforming has been characterized through ensemble-average loss metrics \cite{Love2007,raghavan2013ensemble}. However, these guarantees are inherently probabilistic. Deterministic packings instead provide guaranteed minimum distance and worst-case performance, particularly valuable for moderate codebook sizes \cite{pllaha2022binary}, while their algebraic structure can reduce implementation, storage, and decoding complexity \cite{Heath2009,pllaha2022binary}. \cref{sec:design_isotropic_BGC} exploits these constructions for spatially isotropic Grassmannian beamspace design.

\subsubsection{Distinctions from Limited-Feedback Grassmannian Beamforming}\label{sec:distinction_LFGB_and_BGC}
As \cref{sec:relation_ML_MAP_minimumdist} will show, when channel coefficient $\alpha$ is unknown, the DoA information is contained in one-dimensional subspace $\Span{\W\a_{\SS}(f)}$. The performance of the ML estimator therefore motivates designing the set of possible beamspace responses  $\{ \Span{\W\a_{\SS}(f)} : f\in\mathcal G \}$ as a Grassmannian codebook, which we call a \emph{beamspace Grassmannian code} (BGC). Unlike limited-feedback beamforming, where the {beamformers} can be chosen directly as a Grassmannian packing \cite{love2003grassmanian,mukkavilli2003beamforming}, BGC codewords are the \emph{beamspace responses} $\W\a_{\SS}(f)$. Thus, the induced Grassmannian geometry is  jointly determined by the beamformer $\W$ and the array manifold. 
This raises a central realizability question: \emph{Can $\W$ realize a prescribed Grassmannian line packing for a given array geometry $\SS$?} In \cref{sec:grassmannian_inv}, we show that arbitrary packings are realizable when the number of antennas $N_a$ equals the angular grid size $N_g$, whereas for $N_g>N_a$ realizability becomes array-dependent and arbitrary packings are generally infeasible.

{A second distinction concerns the origin of the Grassmannian criterion. In limited-feedback beamforming, it arises from quantizing a full-CSI beamformer under a statistical channel model. For i.i.d.\ Rayleigh fading, minimizing an upper bound on average SNR loss yields Grassmannian line packing \cite{love2003grassmanian}; outage minimization gives the same criterion for single-stream beamforming \cite{mukkavilli2003beamforming}, while multi-stream precoding can induce different metrics \cite{love2005limited}.} In contrast, BGCs arise directly from the ML estimation objective, with $\alpha$ modeled as an unknown deterministic variable.

Limited-feedback beamforming over correlated channels offers a further point of comparison, since it also modifies Grassmannian codebooks to account for channel structure, but it does so in a fundamentally different manner. Given a known correlation matrix $\mathbf{R}$, prior work adapts Grassmannian codebooks through correlation-dependent skewing or companding followed by normalization \cite{love2006limited}. For RVQ, \cite{raghavan2013ensemble} similarly advocates a fixed skew matrix aligned empirically with dominant channel directions. This differs fundamentally from our BGC design: here, $\W$ must be chosen so that
$\{ \Span{\W\a_{\SS}(f)} : f\in\mathcal G \}$
itself forms the desired Grassmannian codebook. As discussed in \cref{sec:design_isotropic_BGC}, this cannot in general be achieved by a simple pre-(or post)-multiplicative skewing of an existing Grassmannian codebook.

\section{Relation of Subspace Distance Decoding with Maximum-Likelihood 
DoA Decoder
}\label{sec:relation_ML_MAP_minimumdist}

In the absence of noise, it can be observed from \eqref{eq:measurement_model_joint} that the measurement $\y$ lies in the subspace spanned by $\W\a_{\SS}(f)$. In other words, the vectors $\y=\alpha \W\a_{\SS}(f)$ and $\W\a_{\SS}(f)$ are equivalent for any $\alpha \in \CC\backslash\{0\}$ since they represent the same DoA. When $\alpha$ is unknown, we can recover the DoA by identifying the subspace containing $\y$. 
This naturally leads to interpreting the beamformer response $\W\a_{\SS}(f)$ over the grid $f\in\cG$ as a \emph{beamspace Grassmanian code} (BGC).

\begin{definition}[Beamspace Grassmanian Codes] \label{def:BSC}
Given beamformers $\W$ and array geometry $\SS$, beamspace Grassmanian codes corresponding to spatial grid $\cG$ are defined as
\begin{align*}
    \cC(\W,\SS) \triangleq \left\{\Span{\b_n} \text{ s.t } \b_n=\W\a_{\SS}(f_n), f_n\in \cG \right\}.
\end{align*}
\end{definition}

As we will see, {channel sensing performance will depend on} both the beamforming gain and the \emph{minimum subspace distance}, offering new perspectives on the design of the beamforming matrix. {Taking a closer look at the minimum distance 
subspace decoder, we show it to be equivalent to the ML estimator of the DoA when $\alpha$ and $f$ are assumed to be deterministic unknowns. The ML estimator has been extensively studied in the {array processing} literature \cite{stoica1989music,van2002optimum}.}

For brevity of notation, let us define the unit vector corresponding to each codeword as $\widehat{\b}_n \triangleq \frac{\b_n}{\|\b_n\|_2}, n=1,\dots,N_g$. Based on the definition of subspace distance \eqref{eq:subspace_distance_def}, the distance between codewords $\langle\widehat{\b}_l\rangle,\langle\widehat{\b}_k\rangle\!\in\!\cC(\W,\SS)$ is given by $d^{(s)}(\Span{\widehat{\b}_l},\Span{\widehat{\b}_k})\!=\!1\!-\!|\widehat{\b}_l^H\widehat{\b}_k|^2.$
Using \eqref{eq:min_squared_chordal_distance}, the minimum subspace distance of $\cC(\W,\SS)$ is thus defined as \cite{Mahdavifar2022,mahdavifar2024subspace} 
\begin{align*}
     d^{(s)}_{\min}(\cC(\W,\SS)) &= 1-\max_{l,k\in[N_g], l\neq k} |\widehat{\b}_l^H\widehat{\b}_k|^2 \\
     &= 1 - \max_{l,k\in[N_g],l\neq k}\frac{\left|\a_{\SS}(f_l)^H\R_W\a_{\SS}(f_k)\right|^2}{\|\W\a_{\SS}(f_l)\|^2_2\|\W\a_{\SS}(f_k)\|^2_2},
\end{align*}
where $\R_W\!\triangleq\!\W^H\W$. Hence, both beamforming matrix $\W$ and array geometry $\SS$ control the minimum subspace distance.

In presence of noise, given a measurement $\y$ acquired per \eqref{eq:measurement_model_joint} using beamformer $\W$ and array geometry $\SS$, the problem of DoA estimation can be thought of as finding the closest subspace $\langle\widehat{\b}_n\rangle\in  \cC(\W,\SS)$ to the subspace $\langle\y\rangle$. The minimum distance (MD) decoder for BGC $\cC(\W,\SS)$ is thus defined as
\begin{align}
     \hat{f}_\text{md}(\y) &\triangleq \argmin_{f_n\in \cG} d^{(s)}(\Span\y,\Span{\widehat{\b}_n}) = \argmax_{f_n\in \cG}|\y^H\widehat{\b}_n|^2. \label{eq:min_dist_subspace_decoder}
\end{align}
\subsection{{Connection to the Maximum-Likelihood Estimator}} \label{sec:connection_to_ML}
{While the MD decoder is a natural choice for decoding arbitrary Grassmanian codes, one may question its utility in estimating the DoA in place of an approach such as the ML estimator, which is well-studied in the literature \cite{stoica1989music,van2002optimum}.} 
Assuming $\alpha$ and $f$ to be  deterministic unknown parameters, the ML estimator of $f$ is given by
\begin{align}
    \hat{f}_\text{ML}(\y) \triangleq \argmax_{f_n\in \cG} \left(\max_{\alpha\in\CC} \ln p(\y;\alpha,f_n)\right).  \label{eq:max_likelihood_estimator_def}
\end{align}
The following \namecref{lem:MD_ML_equivalence} establishes the equivalence between the MD and ML decoders under the Gaussian noise model.
\begin{lemma} \label{lem:MD_ML_equivalence}
Consider beamforming matrix $\W$, satisfying {design choices \labelcref{d:identifiability,d:unitnorm},} is used to obtain measurements $\y$ according to \eqref{eq:measurement_model_joint} with noise distribution satisfying assumption \labelcref{a:noise}. Then,
the minimum subspace distance decoder $\hat{f}_{\text{md}}$ 
\eqref{eq:min_dist_subspace_decoder} is equivalent to the maximum likelihood estimator $\hat{f}_{\text{ML}}$ 
\eqref{eq:max_likelihood_estimator_def}, i.e.,
    \begin{align*}
        \hat{f}_{\text{md}}(\y) = \hat{f}_{\text{ML}}(\y),\ \forall \y\in \CC^{T}.
    \end{align*}
\end{lemma}

\begin{proof}
Under  {design choice \labelcref{d:unitnorm} and assumption \labelcref{a:noise},} 
the distribution $p(\y;\alpha,f_n)$ of the measurement $\y$  is given by 
\begin{align*}
    p(\y;\alpha,f_n) = \CN(\alpha\W\a_{\SS}(f_n),\sigma^2\I).
\end{align*}
The ML estimate $\hat{f}_\text{ML}$ of $f$ 
can be obtained in terms of the ML estimate of the unknown nuisance parameter $\alpha$ as follows:
\begin{align}
    \hat{f}_\text{ML}(\y) &= \argmax_{f_n\in \cG} \left(\max_{\alpha\in\CC} \ln p(\y;\alpha,f_n)\right) \nonumber\\
     &= \argmin_{f_n\in \cG} \left(\min_{\alpha\in\CC} \|\y - \alpha\b_n\|^2_2 \right). \label{eq:ML_opt_problem}
\end{align}
The solution to inner least-squares problem in \eqref{eq:ML_opt_problem} is given by $\hat{\alpha}(f_n) = \frac{\b_n^H\y}{\|\b_n\|^2_2}$. 
Substitution  into \eqref{eq:ML_opt_problem} and simplifying gives
\begin{align}
    \hat{f}_\text{ML}(\y) & =  \argmin_{f_n\in \cG} \y^H\left(\I - \frac{\b_n\b_n^H}{\|\b_n\|^2_2}\right)\y \nonumber\\
    &=  \argmax_{f_n\in \cG} \frac{|\y^H\b_n|^2}{\|\b_n\|_2^2} = \argmax_{f_n\in \cG}|\y^H\widehat{\b}_n|^2, \label{eq:ML_unknown_alpha}
\end{align}
which is equivalent to $\hat{f}_\text{md}(\y)$ defined in \eqref{eq:min_dist_subspace_decoder}. 
\end{proof}

\begin{remark}
    If $\alpha$ is assumed known, the ML estimator,
    \begin{align}
     \hat{f}_{\text{ML}}(\y) = \argmin_{f_n\in \cG}\|\y-\alpha\b_n\|^2_2,   \label{eq:ML_estimator_known_alpha}
    \end{align}
    corresponds to finding $f_n$ which minimizes the L2 distance between $\y$ and $\alpha\b_n$ rather than the subspace distance between $\langle\y\rangle$ and $\langle\b_n\rangle$. {ML estimators \labelcref{eq:ML_unknown_alpha,eq:ML_estimator_known_alpha} are examples of ``noncoherent'' and ``coherent'' (DoA) decoders, respectively.} 
\end{remark}
\cref{lem:MD_ML_equivalence} demonstrates the equivalence of ML and MD decoders for $f$, when $\alpha$ is unknown. In the following we analyze the error exponent of this decoder in order to gain insights into beamformer design. 

\subsection{Subspace Distance and Decoding Error}\label{sec:ssd_decoding_error_bound}

The following theorem provides a bound on the probability of error for the MD decoder, highlighting its dependence on two key quantities: the effective beamforming gain and the minimum subspace distance. 
\begin{theorem}
\label{thm:ML_decoder_error_exp}
Consider beamforming matrix $\W$, satisfying design choices \labelcref{d:identifiability,d:unitnorm}, is used to obtain measurements $\y$ according to \eqref{eq:measurement_model_joint} with noise distribution satisfying assumption \labelcref{a:noise}. 
Then the probability of error of the MD decoder $\hat{f}_\text{md}(\y)$ given the true direction of arrival $f^\star\in \cG$ is bounded by
\begin{align}
    &\PP({\hat{f}_\text{md}(\y)}\neq f^\star) \leq\nonumber \\
    &\quad {(N_g\!-\!1)\exp\left(-\frac{|\alpha|^2}{2\sigma^2}\|\W\a_{\SS}(f^\star)\|^2_2\left(1\!-\!\sqrt{1\!-\!d^{(s)}_{\min}(\cC)}\right)\right)},\label{eq:ml_ub}
\end{align}
where $\cC = \cC(\W,\SS)$
\end{theorem}
The proof follows standard union-bound arguments and pairwise error analysis under Gaussian noise; details are deferred to \cref{apdx:error_exponent_proof}. Our goal is to interpret the result to guide beamformer design. \cref{thm:ML_decoder_error_exp} shows that $\W$ controls both the {total} beamforming gain $\|\W\a_{\SS}(f^\star)\|_2^2$ and the minimum subspace distance $d^{(s)}_{\min}(\cC(\W,\SS))$ appearing in the MD-decoder error exponent. This 
raises 
fundamental design questions: how are these two quantities coupled, and how should $\W$ balance them? 
As the optimal tradeoff is unknown, {we take initial steps toward addressing these questions next.}




\subsection{Spatially Isotropic Beamformers with No Prior on DoA} \label{sec:isotropic_bf_definition}
{
We first} consider beamformer design {with} 
{\em constant beamforming gain over the entire grid $\mathcal{G}$}. { Such uniformity is desirable to ensure that no gridpoint is disadvantaged in absence of prior knowledge of DoA $f$.}\footnote{With spatially selective beamforming, a DoA may fall in a sidelobe (or even a null), causing severe worst-case performance loss.} We refer to these as {\em spatially isotropic beamformers.}

\begin{definition}[Spatially Isotropic Beamformers] \label{def:spatially_istropic}
A beamforming matrix $\W = [\w_1,\dots,\w_T]^H\in \CC^{T\times N_a}$ is spatially isotropic over grid $\cG$ if it satisfies
{\begin{align*}
    \|\W\a_{\SS}(f_n)\|_2^2 = c_{\W}, \forall f_n\in \cG.
\end{align*}
where $c_{\W}$ is a constant that depends on $\W$. }
\end{definition}

{ This property is satisfied by several commonly used beamforming codebooks such as hierarchical beamforming codebook \cite{Alkhateeb2014,Xia2016,noh2017multi,Javidi2019ie8n}, discrete Fourier transform (DFT) beamforming codebook \cite{zhang2021gridless,he2017codebook}, and overlapped beampattern design \cite{Vucetic2016}. We will show in \cref{sec:design_isotropic_BGC} that channel-coding based beamformers \cite{Yu2025ongrid} and Grassmanian line packing based designs naturally satisfy spatial isotropy.
The full hierarchical beamforming codebook satisfies spatial isotropy since every gridpoint is covered by exactly one beamformer in each layer of the hierarchical codebook. 
However, if only a subset of beamformers are used in a given layer (which can happen in adaptive beam alignment), then isotropy may be violated. In such adaptive sensing cases, the beamforming matrix $\W$ is \emph{adapted} to the problem instance (i.e. a specific DoA), which is distinct from the nonadaptive setting considered herein, where the same beamforming matrix $\W$ is designed to work for any problem instance.} 

By restricting our attention to the class of spatially isotropic beamformers, the beamforming gain term $\|\W\a_{\SS}(f_n)\|_2^2$ in the error exponent \eqref{eq:ml_ub} is fixed to $c_{\W}$. Hence, \eqref{eq:ml_ub}
simplifies to
\begin{align*}
    &\PP({\hat{f}_\text{md}(\y)}\neq f_k) \leq\nonumber \\
    &\quad {(N_g-1)\exp\left(-\frac{|\alpha|^2c_{\W}}{2\sigma^2}\left(1\!-\!\sqrt{1\!-\!d^{(s)}_{\min}(\cC(\W,\SS))}\right)\right)},
\end{align*}
which clearly isolates the role of minimum subspace distance $d^{(s)}_{\min}(\cC(\W,\SS))$ in controlling the error {exponent} of {the MD decoder} across different spatially isotropic beamformer designs. Thus, the beamformer design problem reduces to finding spatially isotropic beamformer $\W$ that \emph{maximizes the minimum subspace distance of $\cC(\W,\SS)$}; {
or equivalently, 
to design $\W$ such that $\cC(\W,\SS)$ forms a Grassmanian 
packing.}

\subsection{{Relation of Channel Sensing and Beam Alignment}}
\label{sec:relation_to_beam_alignment}

{
Before developing spatially isotropic BGCs, 
we distinguish channel sensing from the widely studied mmWave beam-alignment (BA) problem \cite{zhang2017codebook,noh2017multi}. First, channel sensing estimates parameters such as DoAs for communication, as well as sensing, and ISAC. BA typically bypasses explicit estimation and directly selects beamforming vectors for data transmission, often via received-power maximization \cite{zhang2017codebook,hur2013millimeter} or hypothesis testing \cite{noh2017multi}. {Importantly, maximizing received power need not optimize DoA estimation. While highly directional beams favor peak gain, they may provide weak gain/phase variation across nearby angles, whereas estimation-oriented designs deliberately enhance such variations \cite{garcia2018optimal}.} 
Secondly, in channel sensing the DoA is encoded in a subspace (see \cref{sec:relation_ML_MAP_minimumdist}), naturally motivating the Grassmannian packing viewpoint adopted here. Thirdly, BA is typically \emph{adaptive}, relying on hierarchical codebooks \cite{hur2013millimeter,Alkhateeb2014,noh2017multi,zhang2017codebook}. The fundamental difference between the adaptive and nonadaptive cases is that adaptivity 
specializes 
$\W$ 
to the observed DoA (specific problem instance), while nonadaptive sensing requires a single predetermined $\W$ to perform uniformly well over all DoAs. We focus on the latter; adaptive BGCs are left for future work.
}

\section{{Design of Spatially Isotropic Beamformers 
Inducing 
Grassmannian Line Packings}} \label{sec:design_isotropic_BGC}

Motivated by \cref{thm:ML_decoder_error_exp}, when no prior DoA region is available, we  design spatially isotropic beamformers using Grassmannian line packings. Given codebook $G\!=\!\{\Span{\b_i}, i\!\in\![N_g]\}$ with near-optimal minimum subspace distance, define the beamformer response matrix 
 $\mathbf{B}_{G}\!\triangleq\![\b_1,\b_2,\dots,\b_{N_g}]\in \mathbb{C}^{T\times N_g}$. We seek a spatially isotropic $\W$ satisfying
\begin{equation}
\W \A_{\SS} = \B_{G},  \label{eqn:induced_Grassm}
\end{equation}
where $\A_{\SS} \triangleq [\a_{\SS}(f_1),\a_{\SS}(f_2),\ldots, \a_{\SS}(f_{N_g})] \in\CC^{N_a\times N_g} $
is the array-manifold matrix on grid $\mathcal{G}$. Thus, $\W$ itself need not be a Grassmannian packing; rather, its array-shaped spatial responses must \emph{induce} the prescribed codebook $G$.

A solution to \eqref{eqn:induced_Grassm} exists for exists for any arbitrary given $\mathbf{B}_{G}$, if and only if the row space of the Grassmannian code $\B_{G}$ is  a subspace of the row-space of $\A_{\SS}$.
This yields two regimes: (i) $\text{rank}(\A_{\SS})=N_g$, termed a `left-invertible array manifold,'' and (ii) $\text{rank}(\A_{\SS})<N_g$, termed `non-left-invertible array manifold.'' In 
regime (i), the Moore--Penrose pseudoinverse $\A_{\SS}^{\dagger}$ is also a left inverse, so any $\B_{G}$ is realizable via
\begin{align}\label{eq:W_pinv}
\W = \B_{G}\A^\dagger_{\SS}.
\end{align}
We call this regime {\em fully realizable} since one can repurpose existing Grassmanian line packings (or modulated channel codes) to design $\W$.\footnote{As shown in \cref{sec:grassmannian_inv}, under mild conditions on $\B_{G}$, this $\W$ also satisfies \labelcref{d:identifiability,d:unitnorm} and is spatially isotropic.}  When $\text{rank}(\A_{\SS})<N_g$, arbitrary Grassmannian packings are no longer realizable. Rather, the structure of $\B_{G}$ must be compatible with the row space of $\A_{\SS}$. We say that this regime permits ``limited realizability", and is heavily influenced by the array geometry. 
For $\text{rank}(\A_{\SS})=N_g$, we further consider (i) analytic Grassmannian constructions in \cref{sec:grassmannian_inv}, and (ii) modulated channel-code constructions and their induced subspace distances in \cref{sec:channel_code_inv}.

\subsection{Beamformers Inducing Grassmannian Line Packings for Left-Invertible Array Manifolds} \label{sec:grassmannian_inv}
In the sequel, we will assume {\labelcref{a:ula} holds, i.e., we have a ULA ($\SS =\UU$).} 
In this case, it can be easily verified using properties of Vandermonde matrices, that $\mathrm{rank}(\A_{\SS})=N_g$ whenever $N_a = N_g$.\footnote{Although $\mathrm{rank}(\A_{\SS})=N_g$ for \(N_a>N_g\) as well, this corresponds to a coarser angular grid than the number of available antennas and is therefore of lesser interest than the higher resolution setting of $N_a\leq N_g$.}
In this case, $\W$ can be designed to realize any prescribed Grassmannian line packing $\mathbf{B}_{G}$ as 
\begin{align}\label{eq:bf_design_pinv} 
\W=\W_{\B} 
\triangleq \frac{1}{N_a} \mathbf{B}_{G}\A^H_{\UU},
\end{align}
since for a ULA and uniform grid, $\A^\dagger_{\UU} = \frac{1}{N_a}\A^H_{\UU}$ when $N_a = N_g$. 
 A question naturally arises whether the aforementioned beamformer(s) $\W_{\B}$ in \eqref{eq:bf_design_pinv} remains spatially isotropic and has unit norm rows so that it satisfies the design choice \labelcref{d:unitnorm}. The following lemma specifies conditions on $\B_{G}$ under which these properties hold. 

\begin{lemma}\label{lem:unit_circ_spatial_isotropy}
Given a codebook $\B_G\in \CC^{T\times N_g}$ such that $|[\B_G]_{n,m}| = 1,\ \forall n\in[T],m\in [N_g]$ and assuming $N_a=N_g$, the beamformer $\W$ designed according to \eqref{eq:bf_design_pinv} is a spatially isotropic beamformer and satisfies the design choice \labelcref{d:unitnorm}.
\end{lemma}
\begin{proof}
It can easily be seen that $\W$ is spatially isotropic, since for every gridpoint $f_n\in \cG$ the beamforming gain is
\begin{align*}
     \|\W\a_{\UU}(f_n)\|_2^2 
     =  \|\b_n\|_2^2 = T,
\end{align*}
where $\b_n$ is the $n$-th column of $\B_G$. Also, each row 
$\w_t$ 
of $\W\!\in\!\CC^{T\times N_a}$ has 
unit L2 norm as $\B_G$ has unit modulus entries:
\begin{align*}
    \|\w_t\|^2_2 &=\w^H_t\w_t = \frac{1}{N_a^2}{\b}^{(t)}\A_{\UU}^H\A_{\UU}({\b}^{(t)})^H = \frac{1}{N_a}\|{\b}^{(t)}\|_2^2 =1,
\end{align*}
where $\b^{(t)}$ is the $t$-th row of $\B_G$. Hence, \labelcref{d:unitnorm} is satisfied.
\end{proof}

Thus, if the Grassmannian line packing $\B_{G}$ has 
unit-modulus entries, then one can use \eqref{eq:bf_design_pinv} to simply the design of $\W$, 
ensuring that it 
automatically satisfies spatial isotropy and design choice \labelcref{d:unitnorm}. 
Fortunately, several existing Grassmannian line packing constructions with (near)-optimal minimum distance indeed possess this desired unit modulus structure. We review three of them below, highlighting that optimality of their minimum distance depends on appropriate regimes of $T, N_g$ and $N_a$. These will be  used later to study the empirical performance of the induced beamformers {in \cref{sec:numerical}}.



\subsubsection{Singer-difference-set construction}\label{sec:singer_set_BGC_construct} 
In the quadratic regime $N_g = O(T^2)$, Singer difference sets 
provide a constant-modulus construction for $\B$ which achieves the Welch bound \cite{strohmer2003grassmannian,Giannakis2005}.
For $T=q+1$ where $q$ is a power of a prime number and grid size $N_g = T^2-T +1$, the beamformer response matrix $\B_{\rm S}$ is constructed by selecting a subset of rows of the $N_g$-point 
DFT 
matrix indexed by the elements of the Singer difference set $\Omega_{\rm S}$ defined as follows.
\begin{definition}[Singer set \cite{singer1938theorem,Giannakis2005}]\label{def:Singer_array} 
For $T = q+1$ and $N_g=T^2-T+1$ where $q=p^m$ for integer $m>0$ and prime $p$, the \emph{Singer set} $\Omega_{\rm S}$, with cardinality $|\Omega_{\rm S}|=T$, is defined as:
\begin{align*}
  \Omega_{\rm S}=\{d: 0\le d < N_g,\ \mathrm{Tr}(\eta^d)=0\},  
\end{align*}
where $\eta$ is a generator of the multiplicative group of $\mathbb{F}_{q^3}$, and $\mathrm{Tr}:\mathbb{F}_{q^3}\to\mathbb{F}_q$ is the field trace defined by 
$\mathrm{Tr}(\gamma)\triangleq \sum_{i=0}^{2}\gamma^{q^i}.$
\end{definition}
Thus, $\B_{\rm S}$ is given by \cite{strohmer2003grassmannian,Giannakis2005}
\begin{align}
    [\B_{\rm S}]_{t,n} = \exp\left( j\frac{ 2\pi d_t (n-1)}{N_g} \right),\ t\!\in\![T],\ n\!\in\![N_g]. 
    \label{eq:singer_codebook}
\end{align}
where $d_1,\dots,d_{T} \in \Omega_{\rm S}$. The normalized columns $\widehat{\b}_n=\b_n/\sqrt{T}$ of $\B_{\rm S}$ form a Welch-bound-achieving Grassmannian line packing \cite{strohmer2003grassmannian,Giannakis2005} which implies
\begin{align*} 
\left| \widehat{\b}_{\ell}^{H}\widehat{\b}_{k} \right|^2 = \frac{N_g-T}{T(N_g-1)} = \frac{T-1}{T^2}, \qquad \ell\neq k.
\end{align*}
Consequently, for $\W_{\rm S} = \B_{\rm S}\A_{\UU}^{\dagger}$, the minimum subspace distance is given by
\begin{align} 
d_{\min}^{(s)}\bigl( \cC(\W_{\rm S},\UU) \bigr) = 1 - \frac{T-1}{T^2}. 
\label{eq:min_dist_singer_BGC}
\end{align}
Thus, the Singer construction gives an optimal beamformer design in the quadratic regime $N_g=\Theta(T^2)$. 


\subsubsection{Bose-Chowla construction} 
Another construction that nearly achieves the Welch bound in the quadratic regime corresponds to the Bose-Chowla construction of Golomb rulers \cite{bose1962theorems}. Our previous work \cite{mahdavifar2024subspace} examined their use as \emph{antenna-space} subspace codes for spatial sensing. However, their use as BGCs in channel sensing has not been investigated before.
For $T=q$ where $q$ is a prime power and $N_g=T^2-1$, the beamformer response matrix $\B_{\rm BC}$ for Bose-Chowla construction, similar to Singer construction, is given by selecting a subset of rows of the $N_g$-point DFT matrix corresponding to the Bose-Chowla construction, defined as follows.
\begin{definition}[Bose-Chowla construction \cite{bose1962theorems}] \label{def:BC_construct}
Given $q=p^n$ for some prime $p$ and integer $n>0$, let $g$ be a primitive element of $\FF_{q^2}$. Then the Bose-Chowla construction of the Golomb ruler with $q$ elements is given by
\begin{align*}
    \Omega_{\rm BC} \triangleq \left\{ i \in [q^2-2]: g^i - g\in \FF_q\right\}.
\end{align*}
\end{definition}


Thus, $\B_{\rm BC}$ is given by \cite{mahdavifar2024subspace} 
\begin{align}
    [\B_{\rm BC}]_{t,n} = \exp\left( j\frac{ 2\pi d_t (n-1)}{N_g} \right),\ t\!\in\![T],\ n\!\in\![N_g]. 
    \label{eq:BC_codebook}
\end{align}
where $d_1,\dots,d_T\in \Omega_{\rm BC}$. The minimum subspace distance of codebook $\cC_{\rm BC}$, corresponding to the columns of $\B_{\rm BC}$, has been shown to satisfy \cite[Theorem 2]{mahdavifar2024subspace} 
\begin{align*}
    \ds_{\min}(\cC_{\rm BC}) > 1 - \frac{2}{T}. 
\end{align*}
Thus, for $\W_{\rm BC} = \B_{\rm BC} \A_{\UU}^{\dagger}$, the minimum subspace distance of the BGC $\cC(\W_{\rm BC},\UU)$ follows the same lower bound
\begin{align}
    \ds_{\min}(\cC(\W_{\rm BC},\UU)) > 1 - \frac{2}{T}. \label{eq:min_dist_lowerbound_BC}
\end{align}
\cref{eq:min_dist_lowerbound_BC} comes close to the upper bound on $\ds_{\min}$ obtained using the Welch bound which is $\approx 1 - \frac{1}{T}$. Thus, the Bose--Chowla construction provides a near-optimal 
BGC with slightly larger grid size compared to the Singer construction.

\subsubsection{Character-polynomial construction} \label{sec:cp}
To construct larger codebooks, we adopt the character-polynomial (CP) codes of \cite[Definition 7]{Mahdavifar2022}. 
The CP codebook exists for $T=q$ where $q=p^m$ and $p$ prime. 
Note that 
CP codes can be constructed beyond quadratic regime $N_g=\Theta(T^2)$ where optimal Singer and Bose-Chowla constructions do not exist. 

\begin{definition}[Character-Polynomial code \cite{Mahdavifar2022}]
\label{def:cp_code}
Let $q=p^m$ be a prime power and let $1\leq k<q$. Define the polynomial family $\mathcal{F}_{k,q} \triangleq \left\{ g(X)= \sum_{1\leq i\leq k, i\not\equiv 0\;(\mathrm{mod}\;p)} \beta_iX^i \;:\; \beta_i\in\mathbb{F}_q \right\}.$
Let $\chi:\FF_q\rightarrow\CC$ be a fixed nontrivial additive character of $\FF_{q}$, and let $\{\gamma_1,\ldots,\gamma_q\}=\mathbb{F}_q$. For each polynomial $g\in\cF_{k,q}$, define
\begin{align}
    \b_g \triangleq \begin{bmatrix} \chi\!\left(g(\gamma_1)\right) & \cdots & \chi\!\left(g(\gamma_q)\right) \end{bmatrix}^T \in \mathbb{C}^{q}.
\end{align}
The CP code is the Grassmannian line codebook
\begin{align}
    \cC_{\mathrm{CP}}(k,q) \triangleq \left\{ \left\langle\b_g\right\rangle : g\in\mathcal{F}_{k,q} \right\} \subset \Gr_{\CC}(1,q).
\end{align}
and its cardinality is $\left|\mathcal{C}_{\mathrm{CP}}(k,q)\right| = q^{\,k-\lfloor k/p\rfloor}.$
\end{definition}

Since $T=q$, 
CP codebooks exist 
for grid sizes
\begin{align*}
    N_g = T^{\,k-\lfloor k/p\rfloor},
\end{align*}
where $k\in\{1,\dots,T-1\}$ can be chosen to operate in different regimes of $(T,N_g)$. This approach provides a flexible construction that goes beyond the quadratic regime $N_g = \Theta(T^2)$. 
Choosing the beamformer response matrix $\B_{\rm CP}$ whose columns are given by unit-norm columns corresponding to elements in $\cC_{\rm CP}(k,T)$, we obtain the beamformer matrix $\W_{\rm CP} = \B_{\rm CP}\A_{\UU}^{\dagger}$. The minimum subspace distance of BGC $\cC(\W_{\rm CP},\UU)$ is lower bounded by \cite{Mahdavifar2022}
\begin{align}
    \ds_{\min} \bigl( \cC(\W_{\rm CP},\UU) \bigr) &\geq 1-\frac{(k-1)^2}{T}. \label{eq:cp_dmin_T}
\end{align}
Note that the bound is only informative for $k<\sqrt{T}+1$. Since the additive character $\chi:\FF_{q}\to \CC$ maps elements of $\FF_q$ to the complex unit circle in $\CC$, CP codes provide 
constant-modulus construction for $\B$ which leads to spatially isotropic beamformers $\W_{\rm CP}$. 

These constructions are complementary rather than interchangeable {since they exists for different values of $(T,N_g)$}. Singer and Bose–Chowla provide optimal or near-optimal distance in specific quadratic regimes, whereas CP codes cover a broader set of $(T,N_g)$ values with a minimum-distance lower bound. 
%
%
%
%
%
The following theorem characterizes the performance of the beamformers induced by these Grassmannian line packings in their respective operating regimes.

\begin{theorem} \label{thm:Pe_upperbound_isotropic} 
Consider measurement model \eqref{eq:measurement_model_joint} under assumptions \labelcref{a:grid,a:ula,a:noise}, with $N_a=N_g$. Let $q=p^m$ be a prime power. Let $\W_{\mathsf X}$ be the beamformer designed to induce Grassmanian line packing construction $\mathsf X\in\{\mathrm{CP},\mathrm{S},\mathrm{BC}\}$ 
under admissible $(T,N_g)$ regime.
Then, for true DoA $f^\star\in\cG$, the probability of error of the {MD decoder $\hat{f}_\text{md}(\y)$}, is bounded by
\begin{align*}
    &\PP({\hat{f}_\text{md}}\neq f^\star) \leq\\
&\begin{cases}
    {(T^2-T)\exp\left(-\frac{|\alpha|^2}{2\sigma^2}T\left(1-\frac{\sqrt{T-1}}{T}\right)\right)} &, {\sf X}={\rm S}\\
    {(T^2-2)\exp\left(-\frac{|\alpha|^2}{2\sigma^2}T\left(1-\sqrt{\frac{2}{T}}\right)\right)} &, {\sf X}={\rm BC}\\
    {(T^{k-\lfloor k/p\rfloor}-1)\exp\left(-\frac{|\alpha|^2}{2\sigma^2}T\left(1-\frac{k-1}{\sqrt{T}}\right)\right)} &, {\sf X}={\rm CP}(k,T),
\end{cases}
\end{align*}
for grid sizes
\begin{align*}
    N_g = \begin{cases}
        T^2-T+1 &, {\sf X}={\rm S}\\
    T^2-1 &, {\sf X}={\rm BC}\\
    T^{k-\lfloor k/p\rfloor} &, {\sf X}={\rm CP}(k,T)
    \end{cases}
\end{align*}
where $k\in \{1,\dots,T-1\}$.
\end{theorem}
\begin{proof}
The result follows by substituting the corresponding minimum-distance value or lower bound and the spatially isotropic gain
$\|\W_{\mathsf X}\a_{\UU}(f^\star)\|_2^2=T$ into \cref{thm:ML_decoder_error_exp}.
\end{proof}

Together, the Singer, Bose--Chowla, and CP constructions cover a wide range of $(T,N_g)$ regimes with guarantees on the achievable minimum subspace distance. As we will see in \cref{sec:BGC_Ng_greaterthan_Na}, the Singer and Bose--Chowla constructions can also be utilized in the $N_g > N_a$ regime. 
\cref{sec:numerical} 
will numerically study these beamformer designs when $N_g > O(T^2)$.

\subsection{Modulated Channel Code-Induced Beamformers for Left-Invertible Array Manifolds}\label{sec:channel_code_inv}


We now consider beamformers built from modulated channel codes with good Hamming distance. In particular, recent work \cite{Yu2025ongrid} designed nonadaptive mmWave sensing beamformers from Reed--Muller codes mapped to a BPSK constellation. 
{While \cite{Yu2025ongrid} shows that DoA estimation can be related to channel coding when the channel coefficient $\alpha$ is known (or a good estimate is available), we are interested in understanding how this approach performs when $\alpha$ is \emph{unknown}.}
We therefore study the minimum subspace distance of the induced BGC and relate it to the Hamming distance of the underlying code, revealing an important distinction between the two. {Our analysis encompasses any BPSK-modulated channel codes, including (but not restricted to) Reed-Muller codes.}

Each grid point $f_n\in\cG$ is assigned a BPSK codeword $\b_n\in\{-1,+1\}^{T}$, collected in $\B_\text{BPSK}\in\{-1,1\}^{T\times N_g}$. 
For $N_a=N_g$, \eqref{eq:bf_design_pinv} gives the corresponding beamformer
\begin{align} \label{eq:beamformer_design_BPSK}
\W_\text{BPSK} = \frac{1}{N_a}\B_\text{BPSK}\A_{\UU}^H .
\end{align}
By \cref{lem:unit_circ_spatial_isotropy}, such designs are spatially isotropic and satisfy the design criteria \labelcref{d:unitnorm}.
This raises the question of how the minimum subspace distance of $\cC(\W_{\text{BPSK}},\UU)$ relates to the Hamming distance of $\B_{\text{BPSK}}$, and {which constructions of $\B_{\text{BPSK}}$ maximize the minimum subspace distance}.

\subsubsection{Subspace Distance of BPSK Sensing Channel Codes}
Channel codes $\B_{\text{BPSK}}$ are designed for large Hamming distance. {Under known $\alpha$, \cite{Yu2025ongrid} shows DoA sensing then reduces to channel coding over a Binary-Input AWGN channel. When $\alpha$ is unknown, however,} \cref{sec:relation_ML_MAP_minimumdist} showed that the MD decoder is instead governed by the minimum subspace distance of $\cC(\W_{\text{BPSK}},\UU)$. It is thus natural to ask if large minimum Hamming distance implies large minimum subspace distance.

From \eqref{eq:beamformer_design_BPSK}, $\W\a_{\UU}(f_n)=\b_n$, so by \cref{def:BSC}, BGC $\cC(\W,\UU)=\{\langle\b_n\rangle: f_n\in\cG\}$ has minimum subspace distance
\begin{align*}
    d^{(s)}_{\min}(\cC(\W,\UU)) = 1 -\max_{i,j\in [N_g], i\neq j} \frac{1}{T^2}|\b_i^H\b_j|^2,
\end{align*}
which depends on the codewords only through $|\b_i^H\b_j|$. In contrast, 
the Hamming distance 
{of binary codewords $\b_i,\b_j$ is defined} 
$\dham(\b_i,\b_j)\!\triangleq\!\sum_{t=1}^{T}\bbone\{[\b_i]_t\!\neq\![\b_j]_t\}$. The following lemma links {the inner product with the Hamming distance.}
\begin{lemma}\label{lem:bpsk_inner_prod}
Let $\u,\v\in \{-1,1\}^T$ be BPSK modulated binary codewords with Hamming distance $\dham(\u,\v)$, then we have
\begin{align*}
    {\u}^H{\v} &= T - 2\dham(\u,\v).
\end{align*}
\end{lemma}
\begin{proof}
    For codewords $\u,\v\in \{-1,1\}^T$ we can write
\begin{align*}
    \u^H\v=\sum_{\mathclap{i\in [T],u_i=v_i}}
    {u}_i{v}_i\ +\ \sum_{\mathclap{i\in [T],u_i\neq v_i}}
    {u}_i{v}_i=(T\!-\!\dham(\u,\v))\!-\!\dham(\u,\v),
\end{align*}
where the last equality uses ${u}_i{v}_i = 1$ if ${u}_i = {v}_i$ and ${u}_i{v}_i = -1$ otherwise.
\end{proof}
By \Cref{lem:bpsk_inner_prod}, two BPSK codewords are orthogonal when $\dham(\u,\v)=T/2$, so {\em a large Hamming distance need not translate to a large subspace distance}. This is formalized in the following theorem which derives the minimum subspace distance of the Grassmanian line codebook $\cC_{\B} = \{\Span{\b_1},\dots,\Span{\b_{N_g}}\}$ corresponding to any BPSK-modulated codebook $\B_{\text{BPSK}}$ in terms of the minimum and maximum Hamming distances of $\B_{\text{BPSK}}$,
$\dham_{\min}(\B_{\text{BPSK}}) \triangleq \min_{i,j\in N_g,i\neq j} \dham(\b_i,\b_j),$ and $\dham_{\max}(\B_{\text{BPSK}}) \triangleq \max_{i,j\in N_g,i\neq j} \dham(\b_i,\b_j).$ 
{
\begin{theorem} \label{thm:subspace_ham_BPSK}
Consider a BPSK-modulated binary codebook $\B_{\text{BPSK}} = [\b_1,\b_2,\dots,\b_{N_g}] \in \{-1,1\}^{T\times N_g}$ with minimum Hamming distance $\dham_{\min}(\B_{\text{BPSK}})$ and maximum Hamming distance $\dham_{\max}(\B_{\text{BPSK}})$. Then, the minimum subspace distance of the Grassmanian line codebook $\cC_{\B} = \{\Span{\b_1},\dots,\Span{\b_{N_g}}\} \subset \Gr_{\CC}(1,T)$ is given by
\begin{align*}
    &d^{(s)}_{\min}(\cC_{\B})\\
    &= 1 -  \left(\max\left\{1-\frac{2}{T} \dham_{\min}(\B_{\text{BPSK}}),\frac{2}{T} \dham_{\max}(\B_{\text{BPSK}}) - 1\right\} \right)^2\!.
\end{align*}
\end{theorem}
}
\begin{proof}
Using \Cref{lem:bpsk_inner_prod},
\begin{align*}
    &\max_{i,j\in[N_g],i\neq j} \frac{1}{T}|\b_i^H\b_j|
    = \max_{i,j\in[N_g],i\neq j} \frac{1}{T}|T-2\dham(\b_i,\b_j)| \\
    &= \max_{i,j\in[N_g],i\neq j} \max\Big\{\tfrac{1}{T}(T-2\dham(\b_i,\b_j)),\\
    &\qquad \qquad \qquad\tfrac{1}{T}(2\dham(\b_i,\b_j) -T) \Big\} \\
    &=  \max\left\{1-\tfrac{2}{T} \dham_{\min}(\B_{\text{BPSK}}),\tfrac{2}{T} \dham_{\max}(\B_{\text{BPSK}}) -1\right\},
\end{align*}
where the second line uses $|x|=\max\{x,-x\}$. Substituting into $\ds_{\min}(\cC_{\B}) = 1 - \max_{i\neq j}\frac{1}{T^2}|\b_i^H\b_j|^2$ yields the claim.
\end{proof}
\begin{remark}\label{rem:min_subspace_0}
\Cref{thm:subspace_ham_BPSK} shows that the minimum subspace distance depends on {\em both} the minimum and the maximum Hamming distance. It also reveals a sharp contrast with channel coding: two binary codewords are farthest apart when $\dham(\u,\v)\!=\!T$, i.e. $\u\!=\!\mathbf{1}\oplus\v$, but the corresponding BPSK codewords then satisfy $\tilde{\u}\!=\!-\tilde{\v}$ and hence they span the {\em same} line, giving $d^{(s)}(\langle\tilde{\u}\rangle,\langle\tilde{\v}\rangle)\!=\!0$. {This underscores the need to revisit channel-coding-based isotropic beamformer design through the lens of subspace distance:} for $T\!=\!N_g$, the maximal minimum subspace distance of $1$ (orthogonality) requires codewords to be exactly $\frac{T}{2}$ apart in Hamming distance.
\end{remark}

\cref{thm:subspace_ham_BPSK} is illustrated in \cref{fig:dHam_vs_dsmin}, which plots {the minimum subspace distance} $d^{(s)}_{\min}(\cC(\W_{\text{BPSK}},\UU))$ against the 
min. Hamming distance $\dham_{\min}(\B_{\text{BPSK}})/T$ for several values of {the min-to-max distance ratio} $\rho= \dham_{\min}(\B_{\text{BPSK}})/\dham_{\max}(\B_{\text{BPSK}})$. The maximal $d^{(s)}_{\min}$ (circles) is attained, for given $\rho$ and $T$, by balancing the two terms in \cref{thm:subspace_ham_BPSK}, i.e., setting $1-\frac{2}{T} \dham_{\min}(\B_{\text{BPSK}})=\frac{2}{T} \dham_{\max}(\B_{\text{BPSK}}) - 1$, which yields
\begin{align}
     \dham_{\min}(\B_{\text{BPSK}})= T-\dham_{\max}(\B_{\text{BPSK}})=\frac{T}{1+\rho}.\label{eq:opt_ham}
\end{align}
\cref{fig:rho_vs_dsmin} plots this maximal $d^{(s)}_{\min}$ versus $\rho$: only codebooks with $\rho$ close to $1$ achieve a large minimum subspace distance ($d^{(s)}_{\min}=1$ only if $\rho=1$). Extending this relationship to constellations beyond BPSK is 
a key 
direction for future work.
\begin{figure}
\newcommand{\numsamp}{100}
    	\centering
        \subfloat[Subspace distance \emph{vs} (normalized) Hamming distance.
     ]{\label{fig:dHam_vs_dsmin}
    	\begin{tikzpicture}
    		\begin{axis}[width=8.5 cm,height= 3.5 cm,ylabel={$d^{(s)}_{\min}(\cC(\W_{\text{BPSK}},\UU))$},
            xlabel= {$\dham_{\min}(\B_{\text{BPSK}})/T$},
        xmin=0,xmax=1,
       ymin=0,ymax=1.05,
       tick label style={font = \scriptsize},
       label style={font=\scriptsize},
    legend style = {at={(0.5,1.03)},anchor=south,draw=none,fill=none,font=\small},legend columns=4,legend style={/tikz/every even column/.append style={column sep=0.1cm}}
      ]
      \addlegendimage{empty legend}
         \addlegendentry{$\rho$:}
         \foreach \rat/\plotcol/\lstyle in {1/black/dashed,0.5/blue/dashdotted,0.1/red/dotted} {
           \edef\temp{\noexpand\addplot[domain=0:\rat,samples=\numsamp, thick,\plotcol,\lstyle]{1-(max(1-2*x,2*x/\rat-1))^2};}\temp
              \addlegendentryexpanded{$\rat$}
          \edef\temp{\noexpand\addplot[only marks,mark=o,mark size=3,line width=1,\plotcol,forget plot] coordinates {({\rat/(1+\rat)},{1-((1-\rat)/(1+\rat))^2})};}\temp
        }
    		\end{axis}
    	\end{tikzpicture}%
         }\\
         \subfloat[Subspace distance \emph{vs} ratio of min and max Hamming distance.
         ]{\label{fig:rho_vs_dsmin}
         \begin{tikzpicture}
    		\begin{axis}[width=8.5 cm,height=3.5 cm,
            ylabel={ $d^{(s)}_{\min}(\cC(\W_{\text{BPSK}},\UU))$},
            xlabel= { $\rho$},
        xmin=0,xmax=1,
       ymin=0,ymax=1.05,
       tick label style={font = \scriptsize},
       label style={font=\scriptsize},
            legend style = {at={(0.5,1.03)},anchor=south,draw=none,fill=none},legend columns=2,legend style={/tikz/every even column/.append style={column sep=0.1cm}}]
     \addplot[domain=0:1,samples=\numsamp, thick,black]{1-((1-x)/(1+x))^2};
      \foreach \rat/\plotcol/\lstyle in {1/black/solid,0.5/blue/dashed,0.1/red/dotted} {
        \edef\temp{\noexpand\addplot[only marks,mark=o,mark size=3,line width=1,\plotcol,forget plot] coordinates {(\rat,{1-((1-\rat)/(1+\rat))^2})};}\temp
        }
    		\end{axis}
    	\end{tikzpicture}
         }
     \caption{Minimum subspace distance $d^{(s)}_{\min}$ of BPSK codebook (\cref{thm:subspace_ham_BPSK}) with given values of \protect\subref{fig:dHam_vs_dsmin} $\dham_{\min}/T$ and \protect\subref{fig:rho_vs_dsmin} $\rho = \dham_{\min}/\dham_{\max}$.
     The maximum value of $d^{(s)}_{\min}/T$ (circles) decreases with $\rho$, where $d^{(s)}_{\min}=1$ only if $\rho=1$.
     }
    \end{figure}

\subsubsection{BGCs Induced by Reed--Muller Codes} \label{sec:RM_induced_BGC}

\Cref{thm:subspace_ham_BPSK} justifies, from a subspace-distance viewpoint, the necessity of \emph{subselecting} codewords from 
codebooks with large (or small) Hamming distance, and raises the question of which subset yields near-optimal subspace distance. 

To illustrate this, we numerically examine choosing subset of RM codes using various heuristics. \cref{fig:rm_subset} shows the minimum subspace distance of a $16\times N_g$ codebook formed by selecting $N_g$ of the $2048$ columns of a $16\times 2048$ {Reed--Muller} RM$(2,4)$ codebook mapped to BPSK, together with the Welch bound \eqref{eq:welch}. Columns are selected uniformly at random (median of $1000$ trials) and deterministically (the first $N_g$ consecutive columns). Neither strategy reaches the Welch bound, and $d^{(s)}_{\min}$ drops to zero as $N_g$ grows, consistent with \Cref{rem:min_subspace_0}: the codebook contains complementary codewords. For deterministic selection this occurs once more than half the codewords are used ($N_g>1024$); for random selection it occurs much earlier, reflecting the high probability of including a complementary pair. For $240\leq N_g\leq 1024$, the deterministic selection achieves a $\ds_{\min}$ of $0.75$ which can be shown to be the maximum possible $\ds_{\min}$ of any subset (of size $N_g$) of RM(2,4) code. This can be proved using the known weight distribution of RM(2,4) codebook \cite{sloane1970weight} along with the binary Hamming bound \cite{macwilliams1977theory} as shown in \cref{apdx:rm_dmin_proof}.
We note that Reed--Muller codes can be used without pruning when $\alpha$ is known or an estimate thereof is used for DoA estimation \cite{Yu2025ongrid}. When $\alpha$ is unknown, whether resources should be on spent estimating it (using separate pilots) is an open question that we investigate empirically in \cref{sec:numerical}.
\begin{figure}
    	\centering
        \newcommand{\numsamp}{100}
        \newcommand{\cwsize}{16}
        \newcommand{\cbsize}{2048}
        \newcommand{\nummc}{1000}
    	\begin{tikzpicture}
    		\begin{axis}[width=8.5 cm,height=3.5 cm,ylabel={$d^{(s)}_{\min}(\cC(\W_{\text{RM}},\UU))$},xlabel= {\# of codewords, $N_g$},
            xmin={\cwsize},xmax={\cbsize},
    extra x ticks={\cwsize},
    tick label style={font = \scriptsize},
       label style={font=\scriptsize},
	    legend style = {at={(0.5,1.03)},anchor=south,draw=none,fill=none,font=\small},legend columns=2,legend style={/tikz/every even column/.append style={column sep=0.2cm}}
      ]
	             \addplot[green!50!black,very thick,dotted,draw]
            table[x=N,y=dminmedian]{data/RM_subset_sel_\cwsize_\cbsize_rnd_stat_\nummc.dat};
            \addlegendentry{Random}
			 \addplot[blue,thick,draw]
            table[x=N,y=dmin]{data/RM_subset_sel_\cwsize_\cbsize_lin.dat};
            \addlegendentry{Deterministic}
            \addplot[domain=\cwsize:\cbsize,samples=\numsamp, thick,black,dashdotted]{1-(x-\cwsize)/((x-1)*\cwsize)};
            \addlegendentry{Welch bound}
			\addplot[red,very thick,dashed,const plot,mark=none, style={dash pattern=on 5pt off 3pt}]
			coordinates {(16,1) (17,0.9375) (240,0.75) (1025,0) (2048,0)};
			\addlegendentry{
            RM(2,4)-specific upper bound}
    		\end{axis}%
    	\end{tikzpicture}%
	     \caption{Minimum subspace distance $d^{(s)}_{\min}$ of pruned Reed-Muller BPSK codebook. 
         Deterministic pruning achieves the RM(2,4)-specific upper bound on $\ds_{\min}$ (see \cref{apdx:rm_dmin_proof}) for $N_g\geq 240$ and abruptly {falls} to zero for $N_g\geq1025$ as complementary codewords get included in the codebook.}\label{fig:rm_subset}
    \end{figure}

We conclude by noting that so-called Operator Reed--Muller (O-RM) codes \cite{ashikhmin2010grassmannian} provide an alternative route to Grassmannian packings from Reed--Muller codes. 
Rather than modulating binary RM codes to BPSK and pruning, O-RM codes use the Reed--Muller algebraic structure to build projection operators whose range spaces are well-separated subspaces. 
One-dimensional  O-RM codes yield explicit BGCs. CP codes \cite{Mahdavifar2022} (cf. \cref{sec:cp}),
based on Reed--Solomon codes, extend the construction to any prime-power $T=q$ with tunable codebook size. O-RM packings exist only for $T=2^k$ and attain $\ds_{\min}(\cC_{\text{O-RM}})=1/2$ while CP codes maintain $\ds_{\min}(\cC_{\text{CP}})\geq 1/2$ with substantially larger codebooks as $T$ grows \cite{Mahdavifar2022}.




\subsection{Non-Left-Invertible Array Manifolds:  Need for Array-aware Designs} \label{sec:BGC_Ng_greaterthan_Na}
When $\A_{\SS}$ has no left inverse, its row space is a strict subspace of $\mathbb{C}^{N_g}$. Hence, an arbitrary line packing $\B_{G}$ need not be realizable through (\ref{eqn:induced_Grassm}). Instead, \emph{the packing must be compatible with the array geometry}. For a given $\B_{G}$, a beamformer $\W$ satisfying (\ref{eqn:induced_Grassm}) exists if and only if
\begin{equation}
\B_{G}\A_{\SS}^{\dagger}\A_{\SS} = \B_{G}. \label{eqn:SuffProj}
\end{equation}
Equivalently, the rows of $\B_{G}$ must be invariant under projection onto 
{the row space of $\A_{\SS}$,} 
i.e., $\mathrm{row}(\B_{G})\subseteq \mathrm{row}(\A_{\SS})$. Thus, (\ref{eqn:SuffProj}) provides an exact test for whether a Grassmannian line packing is realizable by a given array geometry $\SS$.

This leads to a new structural question: \emph{for which non-left-invertible $\A_{\SS}$ and packings $\B_{G}$ does (\ref{eqn:SuffProj}) hold?} To the best of our knowledge, this problem has not been studied. It exposes a fundamental role of array geometry in realizing Grassmannian beamformers for mmWave sensing, further distinguishing our setting from limited-feedback beamforming and noncoherent communications. We make partial progress in addressing this question by identifying certain classes of Grassmannian packings and array geometries for which (\ref{eqn:SuffProj}) is guaranteed.

\begin{lemma}\label{lem:realizability_NgGreaterThanNa}
Consider beamformer response matrix $\B_{\sf X}$ where ${\sf X} \in \{\mathrm{S}, \mathrm{BC}\}$ with appropriate values of $T$ and $N_g$. Define $\mathrm{Aperture}(\Omega_{\sf X}) \triangleq \max(\Omega_{\sf X})-\min(\Omega_{\sf X}) + 1$ and assume $\mathrm{Aperture}(\Omega_{\sf X}) \leq N_a < N_g$. Then, the beamformer $\W_{\sf X}$ designed {per \eqref{eq:W_pinv} under \labelcref{a:ula}} 
satisfies
\begin{align*}
 \W_{\sf X}\A_{\UU}  = \B_{\sf X} .
\end{align*}
\end{lemma}
\begin{proof}
Recall that beamformer response matrix $\B_{\sf X}$ for $\mathsf{X}\in \{\mathrm{BC},\mathrm{S}\}$ is designed by selecting rows from an $N_g\times N_g$ DFT matrix, 
i.e., from 
the array manifold $\A_{\UU}$ of a ULA sampled on uniform grid $\cG$. Hence, $\B_{\sf X}$ can be written as 
\begin{align} \label{eq:beamform_response_singer_bc}
 \B_{\sf X} = \S_{\Omega_{\sf X}}\A_{\UU},   
\end{align}
where $\A_{\UU}\in \CC^{N_a \times N_g}$ is the ULA array manifold. 
This implies that rows of $\B_{\sf X}$ belong to the row space of $\A_{\UU}$. Thus, the beamformer $\W_{\sf X}$ designed per \eqref{eq:W_pinv} satisfies
\begin{align*}
    \W_{\sf X}\A_{\UU} &= (\B_{\sf X}\A_{\UU}^{\dagger})\A_{\UU} = \S_{\Omega_{\sf X}}\A_{\UU}\A_{\UU}^{\dagger} \A_{\UU} \stackrel{(a)}{=}\S_{\Omega_{\sf X}}\A_{\UU} = \B_{\sf X},
\end{align*}
where $(a)$ follows from $\A_{\UU}$ having full row rank due to the Vandermonde structure and $\A_{\UU}\A_{\UU}^{\dagger} =\I$.
\end{proof}


\begin{remark} \label{rem:cyclic_invariance_BC_singer}
Since, from \eqref{eq:beamform_response_singer_bc}, we have $\B_{\sf X}=\S_{\Omega_{\sf X}}\A_{\UU}$ and $\A_{\UU}$ has full row rank, the corresponding beamformer satisfies
\begin{align}
    \W_{\sf X} = \B_{\sf X}\A_{\UU}^{\dagger} = \S_{\Omega_{\sf X}}\A_{\UU}\A_{\UU}^{\dagger} = \S_{\Omega_{\sf X}} .
    \label{eq:antenna_selection}
\end{align}
Thus, $\W_{\sf X}$ is an antenna-selection beamformer. Writing $\Omega_{\sf X}=\{k_1,\ldots,k_T\}$, the $t$-th beamformer $\w_t$ corresponds to selecting the antenna at location $k_t$. Consequently, the Singer- and Bose--Chowla-based designs can be implemented using low-complexity RF switching hardware, without requiring phase shifters or variable-gain combining. 
Thus \Cref{lem:realizability_NgGreaterThanNa} shows that the Singer- and Bose--Chowla-based response matrices are realizable whenever their row-selection sets fit within the available array size, i.e., \(\mathrm{Aperture}(\Omega_{\sf X})\leq N_a<N_g\). 
One can be further minimize $\mathrm{Aperture}(\Omega_{\sf X})$ to ensure $\mathrm{Aperture}(\Omega_{\sf X})\leq N_a$ using the cyclic shift invariance property of Singer and Bose--Chowla sets, i.e., set $(\Omega_{\sf X}\!+\!K)\ \mathrm{mod}\ N_g$ has equivalent difference set properties to $\Omega_{\sf X}$ for any integer $K$ \cite{singer1938theorem}.
\end{remark}

We conclude by noting that Bose--Chowla and Singer antenna-selection rules can also be emulated when $N_g \gg T^2$ by selecting from an $N_g\times N_g$ DFT matrix the rows indexed by $\Omega_{\sf X}$. Although $\B_G$ remains in the row space of the corresponding $\A_{\UU}$, the resulting BGC no longer inherits the near-optimal Grassmannian line-packing guarantees of the original constructions. Nevertheless, our numerical results (see \cref{sec:numerical}) show strong robustness under an $\epsilon$-recovery criterion \cite{Javidi2019ie8n}, where recovery is considered successful whenever $\hat{f}$ lies within an $\epsilon$-neighborhood of the true DoA. 
For suitable $\epsilon$, Bose--Chowla and Singer arrays maintain high recovery probability across varying grid sizes because erroneous estimates tend to cluster near the true DoA rather than disperse broadly.

This motivates a ``rate-distortion''-type analysis for future research, that captures not only whether an error occurs, but also its geometric severity. Such a framework could expose performance gains hidden by exact-recovery metrics, which penalize near-correct and arbitrarily distant (in angle domain) estimates equally. {We defer a full analysis to future work; see \cref{sec:numerical} for preliminary numerical results.}

\section{{Beamspace Grassmannian}  Codes for Spatially Selective Convolutional Beamspaces} \label{sec:harnessing_conv_beamspace}  
So far we have looked at beamformer designs which have uniform total beamforming gain {in all directions (on grid $\cG$),} 
{being   suitable  when we have no prior knowledge about DoA $f$. However, when such prior knowledge \emph{is} available, it can be advantageous to focus energy via beamforming into regions where
$f$ is known to reside.}  This section investigates beamformer designs which are spatially selective, by focusing on
convolutional beamspaces \cite{Vaidyanathan2020,Vaidyanathan2024,Pote2023novel} and understanding their effectiveness as {BGCs.}
{
Convolutional beamspace methods are particularly attractive since they preserve the harmonic structure of the array manifold after beamspace transformation, enabling classical subspace-based DoA estimation techniques to be applied in beamspace while providing spatial filtering to concentrate gain in desired directions \cite{Vaidyanathan2020,Vaidyanathan2024,Pote2023novel}. Although we focus on convolutional beamspaces as a concrete example, the discussion extends more broadly to beamformer classes whose beamspace response admits independent control over beamforming gain and the subspace distance.
}

\subsection{Review of Convolutional Beamspaces} \label{sec:review_CBS}
The approach of convolutional beamspace \cite{Vaidyanathan2020,Vaidyanathan2024,Pote2023novel} is to construct beamformers $\w_1,\dots,\w_T\!\in\!\CC^{N_a}$ as shifted copies of a {shorter} length-$P\!<\!N_a$ filter $\w\!\in\!\CC^{P}$, suitably zero-padded.

\begin{definition}[Convolutional Beamspace (CB) \cite{Vaidyanathan2020thf,Pote2023novel}] \label{def:conv_beamspace}
Given a length $P$ filter $\w \in \CC^P$, and a set of shifts $\widehat{\SS}=\{k_1,k_2,\dots,k_T\}$ such that $P + \max_{t\in [T]}k_t\leq N_a$, the linear convolutional beamspace $\W_{\text{CB}}(\w,\widehat{\SS}) = [\w_1,\w_2,\dots,\w_T]^H$ is constructed as follows 
\begin{align*}
    \w_t^H = [\zero_{k_t}^H, \w^H,\zero^H_{N_a-P-k_t}], t=1,\dots,T.
\end{align*} 
\end{definition}
Measurements acquired using such a beamforming matrix 
can be written as the convolution of the array output with the filter $\w$ 
followed by decimation corresponding to the indices in $\widehat{\SS}$.
In absence of noise, the measurements are given by
\begin{align}\label{eq:conv_beamspace}
    \y 
    \!=\!\alpha\begin{bmatrix} \w_1^H\a_{\UU}(f)\\ \vdots \\ \w_T^H\a_{\UU}(f) \end{bmatrix}\!=\! \alpha B(f;\w)\begin{bmatrix} e^{j \pi k_1 f}\\ \vdots \\ e^{j \pi k_T f} \end{bmatrix} 
    \!=\! \alpha B(f;\w) \a_{\widehat{\SS}}(f),
\end{align}
where $B(f;\w) \triangleq \sum_{n=1}^{P}w_n^*e^{j\pi (n-1)f}$ is the {(spatial)} frequency response of the filter and $\a_{\widehat{\SS}}(f)$ is the array steering vector with sensor positions given by $\widehat{\SS}$. Hence, after the convolutional beamspace transformation, the measurements maintain the harmonic structure of array manifold $\widehat{\SS}$ while being scaled by filter response $B(f;\w)$. The linear convolutional beamspace approach ensures that \eqref{eq:conv_beamspace} holds over the entire angular domain $f\in [-1,1)$. Hence, the filter length $P$ and maximum possible shift 
{$k_{\max}\triangleq\max \widehat{\SS}$} 
are constrained to satisfy $
k_{\max}+P \leq N_a$ to ensure that all possible shifts in $\widehat{\SS}$ can be realized. 




\subsection{Minimum Subspace Distance of Convolutional {BGCs}
} \label{sec:CBS_min_dist}
A convolutional beamspace has the special property of maintaining the (sub) array manifold structure while also enabling beamforming. We derive a key property of convolutional beamspace subspace codes, which establishes their equivalence to antenna-space {Grassmanian} codes with antenna positions given by shift set $\widehat{\SS}$, regardless of the chosen $\w$.
\begin{lemma}\label{lem:conv_BSC_min_dist}
Given a length $P$ filter $\w\in \CC^P$ with $\|\w\|_2=1$, and a shift pattern $\widehat{\SS}$ such that $
k_{\max}
+P\leq N_a$, the minimum subspace distance of the BGC corresponding to convolutional beamformer $\W_\text{CB}(\w,\widehat{\SS})$ satisfies
\begin{align}
    d^{(s)}_{\min}(\cC(\W_\text{CB}(\w,\widehat{\SS}),\UU)) = d^{(s)}_{\min}(\cC(\I,\widehat{\SS})). \label{eq:min_dist_CB_equalto_ant}
\end{align}    
\end{lemma}
\begin{proof}
To prove \eqref{eq:min_dist_CB_equalto_ant} we show that $\cC(\W_\text{CB}(\w,\widehat{\SS}),\UU)$ is equivalent to $\cC(\I,\widehat{\SS})$  from \cref{def:BSC} as follows:
\begin{align} \label{eq:equivalence_BSC_antenna}
    \cC(\W_\text{CB}(\w,\widehat{\SS}),\UU)
    &\!=\!\left\{\langle \b_n\rangle \text{ s.t } \b_n\!=\!\W_\text{CB}(\w,\widehat{\SS})\a_{\UU}(f_n), f_n\!\in\!\cG \right\}\nonumber\\
    &\stackrel{(a)}{=}\left\{\langle \b_n\rangle \text{ s.t } \b_n\!=\!B(f_n;\w) \a_{\widehat{\SS}}(f_n) , f_n\!\in\!\cG \right\}\nonumber\\
    &\stackrel{(b)}{=} \{ \langle\a_{\widehat{\SS}}(f_n)\rangle, f_n\in \cG\} \nonumber\\
    &= \cC(\I,\widehat{\SS}),
\end{align}
where $(a)$ follows from \eqref{eq:conv_beamspace} and $(b)$ from {scale invariance}, i.e., $\langle B(f_n;\w)\a_{\widehat{\SS}}(f_n)\rangle = \langle\a_{\widehat{\SS}}(f_n)\rangle$. $ \cC(\I,\widehat{\SS})$ is the antenna-space Grassmanian line packing obtained by setting $\W = \I_{T
}$ and array geometry to $\widehat{\SS}$ in \cref{def:BSC}. Since $\cC(\W_\text{CB}(\w,\widehat{\SS}),\UU) = \cC(\I,\widehat{\SS}) $, they have the same minimum subspace distance.
\end{proof}

A direct consequence of this equivalence is that the minimum subspace distance of $\cC(\W,\UU)$ is \emph{invariant to choice of filter} $\w\in \CC^{P}$. Specifically, it only depends on the choice of shifts $\widehat{\SS}$ used to construct the matrix $\W$, while the filter $\w\in \CC^P$ can be chosen freely.
{Hence,} 
an upper bound on the probability of error for convolutional beamspace Grassmanian codes is {obtained} using  \cref{lem:conv_BSC_min_dist,eq:conv_beamspace,thm:ML_decoder_error_exp}:
    \begin{align}
        &\PP({\hat{f}_\text{md}(\y)}\neq f_k)\leq \nonumber \\
         &{(N_g-1)\exp\left(-\frac{|\alpha|^2}{2\sigma^2}|B(f_k;\w)|^2 T \left(1-\sqrt{1 - \ds_{\min}(\cC)}\right)\right)}, \label{eq:CBS_Pe_upperbound}
    \end{align}
where $\ds_{\min}(\cC) = \ds_{\min}(\cC(\I,\widehat{\SS}))$. Note that beamforming gain $|B(f_k;\w)|^2 T$ only depends on the choice of $\w\in \CC^P$ while the minimum subspace distance only depends on the choice of shifts $\widehat{\SS}$. 
%
%
Hence, given a filter $\w$, 
one needs to choose $\widehat{\SS}$ judiciously to maximise the minimum subspace distance. For example, these shifts can be chosen uniformly \cite{Vaidyanathan2020,Pote2023novel} or at random \cite{Myers2025insector}.  
{Next,} 
we study the effect of various choices of shifts and show that the shift pattern corresponding to the Singer and Bose-Chowla constructions are near optimal.


\subsection{ULA shifts vs. Sparse shifts for Convolutional Beamspace} \label{sec:CBS_ula_vs_sparse_shifts}

A common choice of $\widehat{\SS}$ is to choose it uniformly as $\widehat{\SS}=\{0,1,\dots,T-1\}$ \cite{Vaidyanathan2020,Pote2023novel}. 
For $N_g=T^2-1$, using the equivalence between $\cC(\W_{\rm CB}(\w,\widehat{\SS}), \UU)$ and $\cC(\I,\widehat{\SS})$ along with the result on ULA antenna-space subspace codes from \cite{mahdavifar2024subspace} we obtain
\begin{align*}
    d^{(s)}_{\min}\left(\cC \left(\W_{\text{CB}}(\w,\widehat{\SS}_{\rm ULA}),\UU\right) \right) \leq  1 - \frac{4}{\pi^2}.
\end{align*}
This implies that in the $N_g=O(T^2)$ regime, the minimum subspace distance of CBS with  ULA shifts is strictly bounded away from 1 independent of $T$. 

A better shift pattern is obtained using Bose-Chowla {and Singer} sets, 
i.e., {$\widehat{\SS} = \Omega_{\sf X}$, ${\sf X}\in \{\rm S, BC\}$.} 
For $T=q+1$ where $q$ is a prime power, and $N_g = T^2-T+1$, setting the shift pattern $\widehat{\SS} = \Omega_{\rm S}$ and using \cref{lem:conv_BSC_min_dist} along with \eqref{eq:min_dist_singer_BGC} gives
    \begin{align}
        d^{(s)}_{\min}\left(\cC(\W_{\text{CB}}(\w,\widehat{\SS}_{\rm Singer}),\UU)\right) = 1 - \frac{T-1}{T^2}.
        \label{eq:CB_singer_bound}
    \end{align}
Similarly, for $T=q$ and $N_g = T^2-1$, using the Bose--Chowla construction as the shift pattern $\widehat{\SS} = \Omega_{\rm BC}$ gives
\begin{align} 
  d^{(s)}_{\min}\left(\cC(\W_{\text{CB}}(\w,\widehat{\SS}_{\rm BC}),\UU) \right)  \geq 1 - \frac{2}{T}. \label{eq:CB_bosechowla_bound}
\end{align}
Thus, 
the minimum distance using \emph{sparse} shifts, such as {those of the} Bose-Chowla construction, strictly improves over ULA shifts when $1-\frac{2}{T} > 1- \frac{4}{\pi^2}$ which gives $T \geq \left\lceil\frac{\pi^2}{2}\right\rceil = 5$. 
\cref{sec:numerical} will study how the beamforming gain of CBS improves the error probability (over purely spatially isotropic beamformers) when the true DoA lies in a region of interest, keeping the shift pattern fixed while varying beamforming filter $\w \in \CC^{P}$ {(cf. \cref{fig:cbs_results} in \cref{sec:sim_spatially_selective})}.

\section{Numerical Simulations}\label{sec:numerical}

\pgfplotstableread[col sep=comma]{data_final/B1_roi_performance.csv}\BOneData
\pgfplotstablegetelem{0}{ROI_min}\of\BOneData
\edef\BOneROIMin{\pgfplotsretval}
\pgfplotstablegetelem{0}{ROI_max}\of\BOneData
\edef\BOneROIMax{\pgfplotsretval}
\pgfplotstablegetelem{0}{T}\of\BOneData
\edef\BOneT{\pgfplotsretval}
\pgfplotstablegetelem{0}{Na}\of\BOneData
\edef\BOneNa{\pgfplotsretval}
\pgfplotstablegetelem{0}{Ng}\of\BOneData
\edef\BOneNg{\pgfplotsretval}
\pgfplotstablegetelem{0}{P_short}\of\BOneData
\edef\BOnePShort{\pgfplotsretval}
\pgfplotstablegetelem{0}{P_long}\of\BOneData
\edef\BOnePLong{\pgfplotsretval}

We numerically validate the theory and assess BGCs across operating regimes. Unless stated otherwise, we set $f^\star=0.3$ (or the nearest grid point for on-grid experiments) and $\alpha=e^{j\pi/4}$. Results are averaged over $10^5$ independent Monte Carlo trials, with
$\mathrm{SNR}\triangleq 10\log_{10}(\frac{|\alpha|^2}{\sigma_n^2})$. Randomized baselines are ensemble averages over 100 codebook realizations, each using $10^3$ Monte Carlo trials. Fixing $\alpha$ entails no loss of generality: the error probability is phase-invariant, while changes in $|\alpha|$ are absorbed into the SNR definition and hence leave error-versus-SNR curves unchanged. For Bose--Chowla- and Singer-based BGCs, we use the cyclic shift with minimum aperture, as described in \cref{rem:cyclic_invariance_BC_singer}.

\ifdefined\FinalSimulationPlotsMaster
    \let\FinalSimulationPlotEnd\relax
\else
    \documentclass[class=IEEEtran,crop=false]{standalone}
    \IfFileExists{standalone_plot_preamble.inc}
        {
\IfFileExists{mathmacros.sty}{%
    \usepackage{mathmacros}%
}{%
    \IfFileExists{../mathmacros.sty}{%
        \input{../mathmacros.sty}%
    }{%
        \PackageError{standalone_plot_preamble}{mathmacros.sty not found}{%
            Compile from JSAIT_2025_paper or final_simulation_plots.}%
    }%
}
\usepackage{float}
\usepackage{xcolor}
\usepackage[capitalize]{cleveref}
\usepackage{tikz}
\usepackage{pgfplots}
\usepackage{pgfplotstable}
\usepackage[inline]{enumitem}
\pgfplotsset{
    compat=1.16,
    every axis/.append style={
        grid=both,
        minor tick num=1,
        every mark/.append style={solid},
        major grid style={line width=0.25pt,draw=gray!55},
        minor grid style={line width=0.15pt,draw=gray!35}
    }
}
\usepgfplotslibrary{fillbetween,groupplots}
}
        {}
    \begin{document}
    \def\FinalSimulationPlotEnd{\end{document}}
\fi
\IfFileExists{standalone_data_resolver.inc}
    {
\providecommand{\FinalSimulationResolveDataFile}[2]{%
    \IfFileExists{data_final/#1}{%
        \def#2{data_final/#1}%
    }{%
        \IfFileExists{../data_final/#1}{%
            \def#2{../data_final/#1}%
        }{%
            \PackageError{final-simulation-plot}{%
                Required data file #1 not found}{%
                Compile from JSAIT_2025_paper or final_simulation_plots, and
                ensure that JSAIT_2025_paper/data_final/#1 exists.}%
        }%
    }%
}
}
    {}
\FinalSimulationResolveDataFile{A1_theory_bound_empirical.csv}
    {\AOneEmpiricalCSV}
\FinalSimulationResolveDataFile{A1_theory_bound_bounds.csv}
    {\AOneBoundsCSV}
\begin{figure}[H]
    \centering
    \begin{tikzpicture}
    \begin{axis}[
        width=\columnwidth,height=4cm,
        xlabel={SNR (dB)},ylabel={\small $\PP(\hat{f}_{\rm md} \neq f^\star)$},
        xmin=-10,xmax=2,ymin=1e-4,ymax=1,
        ymode=log,unbounded coords=jump,
        legend columns=1,
        legend style={at={(0.02,0.02)},anchor=south west,
        draw=white!60!black,fill=white,font=\scriptsize},
        tick label style={font=\scriptsize},
        label style={font=\scriptsize},
        title style={font=\small},
        grid=both,
        major grid style={line width=0.25pt,draw=gray!55},
        minor grid style={line width=0.18pt,draw=gray!35},
    ]
        \addplot[blue,thick,mark=*] table[col sep=comma,x=SNR,y=Plot_Singer]
            {\AOneEmpiricalCSV};
        \addlegendentry{Empirical}
        \addplot[blue,thick,dashed] table[col sep=comma,x=SNR,y=Plot_UB_Singer]
            {\AOneBoundsCSV};
        \addlegendentry{\cref{thm:ML_decoder_error_exp} bound}
        \addplot[blue,thick,dotted] table[col sep=comma,x=SNR,y=Plot_UB_Singer_MarcumPairwise]
            {\AOneBoundsCSV};
        \addlegendentry{Union bound + Exact PEP}
    \end{axis}
    \end{tikzpicture}
    \caption{{Empirical probability of error vs SNR ($T=32$, $N_a=N_g=993$) 
    for Singer-based BGC. The upper bound in \cref{thm:ML_decoder_error_exp} captures the trend of the empirical plot. 
    }}
    \label{fig:fixed_doa_theorem1_bound}
\end{figure}
\FinalSimulationPlotEnd

\subsection{Spatially Isotropic Case without Prior DoA Knowledge}

\subsubsection{Empirical validation of \cref{thm:ML_decoder_error_exp}}

\Cref{fig:fixed_doa_theorem1_bound} compares the upper bound in \cref{thm:ML_decoder_error_exp} with the empirical probability of error for the BGC induced by the Singer construction for $T=32$, and $N_a=N_g=T^2-T+1=993$. The exponential bound closely captures the decay of the empirical error probability, demonstrating its utility as a tractable beamformer-design criterion for channel sensing. 
A tighter bound follows from the exact PEP in \eqref{eq:exact_pairwise_pep}, which uses the full pairwise subspace-distance spectrum of the BGC.
\subsubsection{Performance of BGC and Channel-code based designs in fully realizable $(N_a=N_g)$ regime}

\ifdefined\FinalSimulationPlotsMaster
    \let\FinalSimulationPlotEnd\relax
\else
    \documentclass[class=IEEEtran,crop=false]{standalone}
    \IfFileExists{standalone_plot_preamble.inc}
        {
\IfFileExists{mathmacros.sty}{%
    \usepackage{mathmacros}%
}{%
    \IfFileExists{../mathmacros.sty}{%
        \input{../mathmacros.sty}%
    }{%
        \PackageError{standalone_plot_preamble}{mathmacros.sty not found}{%
            Compile from JSAIT_2025_paper or final_simulation_plots.}%
    }%
}
\usepackage{float}
\usepackage{xcolor}
\usepackage[capitalize]{cleveref}
\usepackage{tikz}
\usepackage{pgfplots}
\usepackage{pgfplotstable}
\usepackage[inline]{enumitem}
\pgfplotsset{
    compat=1.16,
    every axis/.append style={
        grid=both,
        minor tick num=1,
        every mark/.append style={solid},
        major grid style={line width=0.25pt,draw=gray!55},
        minor grid style={line width=0.15pt,draw=gray!35}
    }
}
\usepgfplotslibrary{fillbetween,groupplots}
}
        {}
    \begin{document}
    \def\FinalSimulationPlotEnd{\end{document}}
\fi

\IfFileExists{data_final/A5_linePacking_collage.csv}{
    \def\AFiveData{data_final/A5_linePacking_collage.csv}
}{
    \IfFileExists{../data_final/A5_linePacking_collage.csv}{
        \def\AFiveData{../data_final/A5_linePacking_collage.csv}
    }{
        \PackageError{fig_A5_linePacking_collage}{
            Production CSV A5_linePacking_collage.csv not found}{
            Run fig_A5_linePacking_collage.m in full mode.}
        \def\AFiveData{data_final/A5_linePacking_collage.csv}
    }
}

\IfFileExists{data_final/A5_linePacking_collage_numerical_dmin.csv}{
    \def\AFiveDminCSV{data_final/A5_linePacking_collage_numerical_dmin.csv}
}{
    \IfFileExists{../data_final/A5_linePacking_collage_numerical_dmin.csv}{
        \def\AFiveDminCSV{../data_final/A5_linePacking_collage_numerical_dmin.csv}
    }{
        \PackageError{fig_A5_linePacking_collage}{
            Numerical minimum-distance CSV not found}{
            Run fig_A5_linePacking_collage.m in full mode.}
        \def\AFiveDminCSV{data_final/A5_linePacking_collage_numerical_dmin.csv}
    }
}

\pgfplotstableread[col sep=comma]{\AFiveDminCSV}\AFiveDminData
\pgfplotstablegetelem{0}{dmin_RM_numerical}\of\AFiveDminData
\edef\AFiveDminA{\pgfplotsretval}
\pgfplotstablegetelem{2}{dmin_BC_numerical}\of\AFiveDminData
\edef\AFiveDminB{\pgfplotsretval}
\pgfplotstablegetelem{1}{dmin_Singer_numerical}\of\AFiveDminData
\edef\AFiveDminC{\pgfplotsretval}
\pgfplotstablegetelem{3}{dmin_CP_numerical}\of\AFiveDminData
\edef\AFiveDminD{\pgfplotsretval}
\pgfplotstablegetelem{4}{dmin_RM_numerical}\of\AFiveDminData
\edef\AFiveDminE{\pgfplotsretval}
\pgfplotstablegetelem{6}{dmin_BC_numerical}\of\AFiveDminData
\edef\AFiveDminF{\pgfplotsretval}
\pgfplotstablegetelem{5}{dmin_Singer_numerical}\of\AFiveDminData
\edef\AFiveDminG{\pgfplotsretval}
\pgfplotstablegetelem{7}{dmin_CP_numerical}\of\AFiveDminData
\edef\AFiveDminH{\pgfplotsretval}
\newcommand{\AFivePrintDmin}[1]{%
    \pgfmathprintnumber[fixed,fixed zerofill,precision=3]{#1}}

\pgfplotsset{
    a5 cp/.style={blue,line width=0.70pt,solid,mark=diamond*,mark size=1.7pt,
        mark repeat=4,mark phase=1,
        mark options={solid,fill=blue}},
    a5 bose chowla/.style={red,line width=0.70pt,solid,mark=square*,mark size=1.7pt,
        mark repeat=4,mark phase=1,
        mark options={solid,fill=red}},
    a5 singer/.style={green!55!black,line width=0.70pt,solid,mark=*,mark size=1.7pt,
        mark repeat=4,mark phase=1,
        mark options={solid,fill=green!55!black}},
    a5 rm/.style={orange!85!black,line width=0.70pt,solid,mark=triangle*,mark size=1.8pt,
        mark repeat=4,mark phase=1,
        mark options={solid,fill=orange!85!black}},
    a5 rvq/.style={black,line width=0.70pt,solid,mark=o,
        mark repeat=4,mark phase=2,
        mark size=2.4pt,
        mark options={solid,fill=white,draw=black,line width=0.85pt}},
    a5 random phase/.style={violet!90!black,line width=0.65pt,dashed},
    a5 random antenna selection/.style={cyan!80!blue,line width=0.70pt,
        mark repeat=4,mark phase=4,
        solid,mark=asterisk,mark size=3.0pt,
        mark options={solid,draw=cyan!80!blue,line width=1.0pt}}
}

\begin{figure*}[t]
\centering
\resizebox{\textwidth}{!}{%
\begin{tikzpicture}
\begin{groupplot}[
    group style={group size=4 by 2,horizontal sep=0.30cm,vertical sep=1.10cm},
    width=4.25cm,height=3.25cm,
    ymin=1e-4,ymax=1,ymode=log,
    ytick={1e-4,1e-3,1e-2,1e-1,1},minor y tick num=9,
    grid=both,
    major grid style={line width=0.25pt,draw=gray!55},
    minor grid style={line width=0.18pt,draw=gray!35},
    unbounded coords=jump,scaled ticks=false,
    tick label style={font=\tiny},
    label style={font=\tiny},
    title style={font=\tiny,align=center,
        at={(axis description cs:0.5,0.9)},anchor=south}
]

\nextgroupplot[title={(a) RM$(2,3)$, $d^{(s)}_{\text{min}}=\AFivePrintDmin{\AFiveDminA}$\\
    $T=8$, $N_a=N_g=64$},
    xmin=-5,xmax=10,xlabel={},
    legend to name=AFiveCommonLegend,legend columns=7,
    legend style={draw=white!60!black,fill=white,font=\scriptsize,
        /tikz/every even column/.append style={column sep=0.12cm}}]
    \addlegendimage{a5 rm}\addlegendentry{RM(subspace)}
    \addlegendimage{a5 singer}\addlegendentry{Singer}
    \addlegendimage{a5 bose chowla}\addlegendentry{Bose--Chowla}
    \addlegendimage{a5 cp}\addlegendentry{CP}
    \addlegendimage{a5 rvq}\addlegendentry{RVQ}
    \addlegendimage{a5 random phase}\addlegendentry{Random phase}
    \addlegendimage{a5 random antenna selection}\addlegendentry{Random AS}
    \addplot[a5 rvq,forget plot] table[col sep=comma,x=A_SNR,y=A_Plot_RVQ]{\AFiveData};
    \addplot[a5 random phase,forget plot] table[col sep=comma,x=A_SNR,y=A_Plot_RandomPhase]{\AFiveData};
    \addplot[a5 random antenna selection,forget plot] table[col sep=comma,x=A_SNR,y=A_Plot_RandomAntennaSelection]{\AFiveData};
    \addplot[a5 rm,forget plot] table[col sep=comma,x=A_SNR,y=A_Plot_RM]{\AFiveData};

\nextgroupplot[title={(b) Bose--Chowla, $d^{(s)}_{\text{min}}=\AFivePrintDmin{\AFiveDminB}$,\\ $T=8$, $N_a=N_g=63$},
    xmin=-5,xmax=10,xlabel={},yticklabels=\empty]
    \addplot[a5 rvq,forget plot] table[col sep=comma,x=E_SNR,y=E_Plot_RVQ]{\AFiveData};
    \addplot[a5 random phase,forget plot] table[col sep=comma,x=E_SNR,y=E_Plot_RandomPhase]{\AFiveData};
    \addplot[a5 random antenna selection,forget plot] table[col sep=comma,x=E_SNR,y=E_Plot_RandomAntennaSelection]{\AFiveData};
    \addplot[a5 bose chowla,forget plot] table[col sep=comma,x=E_SNR,y=E_Plot_BC]{\AFiveData};

\nextgroupplot[title={(c) Singer, $d^{(s)}_{\text{min}}=\AFivePrintDmin{\AFiveDminC}$\\
    $T=9$, $N_a=N_g=73$},
    xmin=-5,xmax=10,xlabel={},yticklabels=\empty]
    \addplot[a5 rvq,forget plot] table[col sep=comma,x=C_SNR,y=C_Plot_RVQ]{\AFiveData};
    \addplot[a5 random phase,forget plot] table[col sep=comma,x=C_SNR,y=C_Plot_RandomPhase]{\AFiveData};
    \addplot[a5 random antenna selection,forget plot] table[col sep=comma,x=C_SNR,y=C_Plot_RandomAntennaSelection]{\AFiveData};
    \addplot[a5 singer,forget plot] table[col sep=comma,x=C_SNR,y=C_Plot_Singer]{\AFiveData};

\nextgroupplot[title={(d) CP, $d^{(s)}_{\text{min}}=\AFivePrintDmin{\AFiveDminD}$\\ $T=9$, $N_a=N_g=81$},
    xmin=-5,xmax=10,xlabel={},yticklabels=\empty]
    \addplot[a5 rvq,forget plot] table[col sep=comma,x=G_SNR,y=G_Plot_RVQ]{\AFiveData};
    \addplot[a5 random phase,forget plot] table[col sep=comma,x=G_SNR,y=G_Plot_RandomPhase]{\AFiveData};
    \addplot[a5 random antenna selection,forget plot] table[col sep=comma,x=G_SNR,y=G_Plot_RandomAntennaSelection]{\AFiveData};
    \addplot[a5 cp,forget plot] table[col sep=comma,x=G_SNR,y=G_Plot_CP]{\AFiveData};

\nextgroupplot[title={(e) RM$(2,4)$,$d^{(s)}_{\text{min}}=\AFivePrintDmin{\AFiveDminE}$\\ $T=16$, $N_a=N_g=1024$},
    xmin=-10,xmax=5,xlabel={SNR (dB)}]
    \addplot[a5 rvq,forget plot] table[col sep=comma,x=B_SNR,y=B_Plot_RVQ]{\AFiveData};
    \addplot[a5 random phase,forget plot] table[col sep=comma,x=B_SNR,y=B_Plot_RandomPhase]{\AFiveData};
    \addplot[a5 random antenna selection,forget plot] table[col sep=comma,x=B_SNR,y=B_Plot_RandomAntennaSelection]{\AFiveData};
    \addplot[a5 rm,forget plot] table[col sep=comma,x=B_SNR,y=B_Plot_RM]{\AFiveData};

\nextgroupplot[title={(f) Bose--Chowla, $d^{(s)}_{\text{min}}=\AFivePrintDmin{\AFiveDminF}$\\ $T=16$, $N_a=N_g=255$},
    xmin=-10,xmax=5,xlabel={SNR (dB)},yticklabels=\empty]
    \addplot[a5 rvq,forget plot] table[col sep=comma,x=F_SNR,y=F_Plot_RVQ]{\AFiveData};
    \addplot[a5 random phase,forget plot] table[col sep=comma,x=F_SNR,y=F_Plot_RandomPhase]{\AFiveData};
    \addplot[a5 random antenna selection,forget plot] table[col sep=comma,x=F_SNR,y=F_Plot_RandomAntennaSelection]{\AFiveData};
    \addplot[a5 bose chowla,forget plot] table[col sep=comma,x=F_SNR,y=F_Plot_BC]{\AFiveData};

\nextgroupplot[title={(g) Singer, $d^{(s)}_{\text{min}}=\AFivePrintDmin{\AFiveDminG}$\\ $T=17$, $N_a=N_g=273$},
    xmin=-10,xmax=5,xlabel={SNR (dB)},yticklabels=\empty]
    \addplot[a5 rvq,forget plot] table[col sep=comma,x=D_SNR,y=D_Plot_RVQ]{\AFiveData};
    \addplot[a5 random phase,forget plot] table[col sep=comma,x=D_SNR,y=D_Plot_RandomPhase]{\AFiveData};
    \addplot[a5 random antenna selection,forget plot] table[col sep=comma,x=D_SNR,y=D_Plot_RandomAntennaSelection]{\AFiveData};
    \addplot[a5 singer,forget plot] table[col sep=comma,x=D_SNR,y=D_Plot_Singer]{\AFiveData};

\nextgroupplot[title={(h) CP, $d^{(s)}_{\text{min}}=\AFivePrintDmin{\AFiveDminH}$\\ $T=17$, $N_a=N_g=289$},
    xmin=-10,xmax=5,xlabel={SNR (dB)},yticklabels=\empty]
    \addplot[a5 rvq,forget plot] table[col sep=comma,x=H_SNR,y=H_Plot_RVQ]{\AFiveData};
    \addplot[a5 random phase,forget plot] table[col sep=comma,x=H_SNR,y=H_Plot_RandomPhase]{\AFiveData};
    \addplot[a5 random antenna selection,forget plot] table[col sep=comma,x=H_SNR,y=H_Plot_RandomAntennaSelection]{\AFiveData};
    \addplot[a5 cp,forget plot] table[col sep=comma,x=H_SNR,y=H_Plot_CP]{\AFiveData};

\end{groupplot}
\node[rotate=90,anchor=south,font=\scriptsize] at
    ($(group c1r1.west)!0.5!(group c1r2.west)+(-0.80cm,0)$)
    {$\PP(\hat{f}_{\rm md} \neq f^*)$};
\end{tikzpicture}
}%
\par
\ref{AFiveCommonLegend}

\caption{
Comparison of deterministic BGC constructions 
with randomized baselines (on-grid DoA, $N_a=N_g$). While the gap reduces for larger $(T,N_a=N_g)$, randomized methods lack deterministic guarantees.
}
\label{fig:final_line_packing_collage}
\end{figure*}
\FinalSimulationPlotEnd


\ifdefined\FinalSimulationPlotsMaster
    \let\FinalSimulationPlotEnd\relax
\else
    \documentclass[class=IEEEtran,crop=false]{standalone}
    \IfFileExists{standalone_plot_preamble.inc}
        {
\IfFileExists{mathmacros.sty}{%
    \usepackage{mathmacros}%
}{%
    \IfFileExists{../mathmacros.sty}{%
        \input{../mathmacros.sty}%
    }{%
        \PackageError{standalone_plot_preamble}{mathmacros.sty not found}{%
            Compile from JSAIT_2025_paper or final_simulation_plots.}%
    }%
}
\usepackage{float}
\usepackage{xcolor}
\usepackage[capitalize]{cleveref}
\usepackage{tikz}
\usepackage{pgfplots}
\usepackage{pgfplotstable}
\usepackage[inline]{enumitem}
\pgfplotsset{
    compat=1.16,
    every axis/.append style={
        grid=both,
        minor tick num=1,
        every mark/.append style={solid},
        major grid style={line width=0.25pt,draw=gray!55},
        minor grid style={line width=0.15pt,draw=gray!35}
    }
}
\usepgfplotslibrary{fillbetween,groupplots}
}
        {}
    \begin{document}
    \def\FinalSimulationPlotEnd{\end{document}}
\fi

\IfFileExists{data_final/A5_BC_RM_random_T16_Ng255.csv}{
    \def\AFiveSingleData{data_final/A5_BC_RM_random_T16_Ng255.csv}
}{
    \IfFileExists{../data_final/A5_BC_RM_random_T16_Ng255.csv}{
        \def\AFiveSingleData{../data_final/A5_BC_RM_random_T16_Ng255.csv}
    }{
        \PackageError{fig_A5_BC_RM_random_T16_Ng255}{
            Production CSV A5_BC_RM_random_T16_Ng255.csv not found}{
            Run fig_A5_BC_RM_random_T16_Ng255.m in full mode.}
        \def\AFiveSingleData{data_final/A5_BC_RM_random_T16_Ng255.csv}
    }
}

\IfFileExists{data_final/A5_linePacking_collage_numerical_dmin.csv}{
    \def\AFiveDminCSV{data_final/A5_linePacking_collage_numerical_dmin.csv}
}{
    \IfFileExists{../data_final/A5_linePacking_collage_numerical_dmin.csv}{
        \def\AFiveDminCSV{../data_final/A5_linePacking_collage_numerical_dmin.csv}
    }{
        \PackageError{fig_A5_linePacking_collage}{
            Numerical minimum-distance CSV not found}{
            Run fig_A5_linePacking_collage.m in full mode.}
        \def\AFiveDminCSV{data_final/A5_linePacking_collage_numerical_dmin.csv}
    }
}

\pgfplotstableread[col sep=comma]{\AFiveDminCSV}\AFiveDminData
\pgfplotstablegetelem{0}{dmin_RM_numerical}\of\AFiveDminData
\edef\AFiveDminA{\pgfplotsretval}
\pgfplotstablegetelem{2}{dmin_BC_numerical}\of\AFiveDminData
\edef\AFiveDminB{\pgfplotsretval}
\pgfplotstablegetelem{1}{dmin_Singer_numerical}\of\AFiveDminData
\edef\AFiveDminC{\pgfplotsretval}
\pgfplotstablegetelem{3}{dmin_CP_numerical}\of\AFiveDminData
\edef\AFiveDminD{\pgfplotsretval}
\pgfplotstablegetelem{4}{dmin_RM_numerical}\of\AFiveDminData
\edef\AFiveDminE{\pgfplotsretval}
\pgfplotstablegetelem{6}{dmin_BC_numerical}\of\AFiveDminData
\edef\AFiveDminF{\pgfplotsretval}
\pgfplotstablegetelem{5}{dmin_Singer_numerical}\of\AFiveDminData
\edef\AFiveDminG{\pgfplotsretval}
\pgfplotstablegetelem{7}{dmin_CP_numerical}\of\AFiveDminData
\edef\AFiveDminH{\pgfplotsretval}
\newcommand{\AFivePrintDmin}[1]{%
    \pgfmathprintnumber[fixed,fixed zerofill,precision=3]{#1}}

\pgfplotstableread[col sep=comma]
    {\AFiveSingleData}\AFiveSingleTable
\pgfplotstablegetelem{0}{T}\of\AFiveSingleTable
\edef\AFiveSingleT{\pgfplotsretval}
\pgfplotstablegetelem{0}{Na}\of\AFiveSingleTable
\edef\AFiveSingleNa{\pgfplotsretval}
\pgfplotstablegetelem{0}{Ng}\of\AFiveSingleTable
\edef\AFiveSingleNg{\pgfplotsretval}
\pgfplotstablegetelem{0}{Nmc}\of\AFiveSingleTable
\edef\AFiveSingleNmc{\pgfplotsretval}
\pgfplotstablegetelem{0}{alpha_abs}\of\AFiveSingleTable
\edef\AFiveSingleAlphaAbs{\pgfplotsretval}
\pgfplotstablegetelem{0}{alpha_phase_rad}\of\AFiveSingleTable
\edef\AFiveSingleAlphaPhase{\pgfplotsretval}
\pgfplotstablegetelem{0}{error_metric_plot_label}\of\AFiveSingleTable
\edef\AFiveSingleErrorMetric{\pgfplotsretval}
\pgfplotstablegetelem{0}{N_random_codebooks}\of\AFiveSingleTable
\edef\AFiveSingleNRandom{\pgfplotsretval}

\pgfplotsset{
    cp/.style={blue,line width=1pt,solid,mark=diamond*,mark size=1.7pt,
        mark repeat=4,mark phase=1,
        mark options={solid,fill=blue}},
    bose chowla/.style={red,line width=1pt,solid,mark=square*,mark size=1.7pt,
        mark repeat=4,mark phase=1,
        mark options={solid,fill=red}},
    singer/.style={green!55!black,line width=1pt,solid,mark=*,mark size=1.7pt,
        mark repeat=4,mark phase=1,
        mark options={solid,fill=green!55!black}},
    rm/.style={orange!85!black,line width=1pt,solid,mark=triangle*,mark size=1.8pt,
        mark repeat=4,mark phase=1,
        mark options={solid,fill=orange!85!black}},
    rvq/.style={black,line width=1pt,solid,mark=o,
        mark repeat=4,mark phase=2,
        mark size=2.4pt,
        mark options={solid,fill=white,draw=black,line width=0.85pt}},
    random phase/.style={violet!90!black,line width=1pt,dashed},
    random antenna selection/.style={cyan!80!blue,line width=1pt,
        mark repeat=4,mark phase=4,
        solid,mark=asterisk,mark size=3.0pt,
        mark options={solid,draw=cyan!80!blue,line width=1.0pt}}
}

\begin{figure}[t]
\centering
\begin{tikzpicture}
\begin{axis}[
    width=0.98\columnwidth,height=4.5cm,
    xlabel={SNR (dB)},ylabel={$\PP(\hat{f}_{\rm md} \neq f^*)$},
    xmin=-10,xmax=5,ymin=1e-4,ymax=1,
    ymode=log,unbounded coords=jump,
    xtick={-10,-5,0,5},
    ytick={1e-4,1e-3,1e-2,1e-1,1},minor y tick num=9,
    grid=both,
    major grid style={line width=0.25pt,draw=gray!55},
    minor grid style={line width=0.18pt,draw=gray!35},
    tick label style={font=\scriptsize},
    label style={font=\scriptsize},title style={font=\small},
    legend columns=1,
    legend style={
        at={(0.02,0.02)},anchor=south west,
        draw=white!60!black,fill=white,font=\scriptsize,
        /tikz/every even column/.append style={column sep=0.12cm}}
]
     \addplot[bose chowla]
        table[x=SNR,y=Plot_BC]{\AFiveSingleTable};
        \addlegendentry{Bose--Chowla, $d^{(s)}_{\text{min}}=\AFivePrintDmin{\AFiveDminF}$}
    \addplot[rm]
        table[x=SNR,y=Plot_RM]{\AFiveSingleTable};
        \addlegendentry{RM(subspace), $d^{(s)}_{\text{min}}=\AFivePrintDmin{\AFiveDminE}$}
    \addplot[rvq]
        table[x=SNR,y=Plot_RVQ]{\AFiveSingleTable};
        \addlegendentry{RVQ}
    \addplot[random phase]
        table[x=SNR,y=Plot_RandomPhase]{\AFiveSingleTable};
        \addlegendentry{Random phase}
    \addplot[random antenna selection]
        table[x=SNR,y=Plot_RandomAntennaSelection]{\AFiveSingleTable};
        \addlegendentry{Random AS}
   
\end{axis}
\end{tikzpicture}
\caption{
On-grid probability of error comparison of Bose--Chowla BGC and BPSK-mapped RM subspace code for $T=16$ and $N_a=N_g=255$, with a uniformly random on-grid DoA. 
The Bose--Chowla BGC outperforms both the optimally pruned RM subspace code and the randomized  baselines due to its superior $d^{(s)}_{\min}$
}
\label{fig:line_packing_BC_RM_T16_Ng255}
\end{figure}
\FinalSimulationPlotEnd

\cref{fig:final_line_packing_collage} compares BGC constructions across several $(T,N_g)$ settings with $N_g=N_a$, so that $\A_{\UU}$ is left-invertible. The true $f^\star$ is drawn uniformly from $\cG$. 
We evaluate the BGCs of \cref{sec:design_isotropic_BGC} at their admissible $(T,N_g)$ pairs. Since Bose--Chowla, Singer, and CP 
constructions do not share common admissible parameters under their theoretical guarantees, they are shown in separate subplots. The Reed--Muller BGC is constructed as in \cref{sec:RM_induced_BGC}, with only the complementary codewords removed to ensure nonzero minimum subspace distance. 
We compare these deterministic constructions against: \begin{enumerate*}[label=(\roman*)]
\item \emph{Random Vector quantization (RVQ)} \cite{Love2007,raghavan2013ensemble}, where columns of $\B_{\rm RVQ}$ are drawn independently and uniformly from the sphere $\{\x\in \mathbb{C}^{T}, \|\x\|^2_2 = T\}$, followed by beamformer design via \eqref{eq:bf_design_pinv},
\item \emph{Random-phase} beamforming with $[\W_{\rm RP}]_{t,n} = \frac{1}{\sqrt{N_a}}e^{j\phi_{t,n}},  \phi_{t,n}\stackrel{\mathrm{i.i.d.}}{\sim} \mathrm{Uniform}([0,2\pi])$  
\cite{marzi2016compressive,alkhateeb2015compressed,venugopal2017channel}), and
\item \emph{Random Antenna Selection (AS)}, with a $T$-subset of antennas drawn uniformly without replacement.
\end{enumerate*}

We consider $T\in \{8,9\}$ and $T \in \{16, 17\}$, according to the admissible BGC parameters. Bose--Chowla and Singer, which are near-optimal and optimal in their respective regimes, consistently perform best for their respective $(T,N_g)$. CP-based BGCs also outperform randomized baselines in the quadratic $N_g=O(T^2)$ regime. For $(T,N_g)=(8,64)$, the BPSK-mapped Reed--Muller BGC outperforms all randomized baselines; for $(T,N_g)=(16,1024)$, where the codebook size is four times $T^2$, it remains competitive with RVQ and random antenna selection. The RVQ gap narrows as codebook size grows, although it should be noted that its performance is ensemble-averaged and offers no guarantee for any particular realization. Deterministic BGCs instead provide explicit $d_{\min}^{(s)}$ guarantees when available, along with additional benefits pertaining to storage and decoding efficiencies \cite{Mahdavifar2024,ashikhmin2010grassmannian}.

 While \cref{fig:final_line_packing_collage} evaluated the full BPSK-mapped RM subspace code (after removing complementary codewords), \cref{fig:line_packing_BC_RM_T16_Ng255} compares a near-optimal Grassmannian BGC with a pruned RM codebook in the same $(T,N_g)$ regime. Motivated by its strong performance in \cref{fig:final_line_packing_collage}, we use the Bose--Chowla construction for the Grassmannian BGC. The RM subspace codebook is formed by selecting $255$ codewords from the $2048$ codewords of $\mathrm{RM}(2,4)$ with the maximum possible $\ds_{\min}$ of $0.75$ after BPSK mapping.\footnote{The optimality of this pruning follows from known properties of the $\mathrm{RM}(2,4)$ weight distribution as shown in \cref{apdx:rm_dmin_proof}.}
The superior probability of error of the Bose--Chowla BGC in \cref{fig:line_packing_BC_RM_T16_Ng255} is consistent with its larger minimum subspace distance compared to RM(subspace) BGC.
\subsubsection{Performance of BGCs for off-grid 
DoA}

\ifdefined\FinalSimulationPlotsMaster
    \let\FinalSimulationPlotEnd\relax
\else
    \documentclass[class=IEEEtran,crop=false]{standalone}
    \IfFileExists{standalone_plot_preamble.inc}
        {
\IfFileExists{mathmacros.sty}{%
    \usepackage{mathmacros}%
}{%
    \IfFileExists{../mathmacros.sty}{%
        \input{../mathmacros.sty}%
    }{%
        \PackageError{standalone_plot_preamble}{mathmacros.sty not found}{%
            Compile from JSAIT_2025_paper or final_simulation_plots.}%
    }%
}
\usepackage{float}
\usepackage{xcolor}
\usepackage[capitalize]{cleveref}
\usepackage{tikz}
\usepackage{pgfplots}
\usepackage{pgfplotstable}
\usepackage[inline]{enumitem}
\pgfplotsset{
    compat=1.16,
    every axis/.append style={
        grid=both,
        minor tick num=1,
        every mark/.append style={solid},
        major grid style={line width=0.25pt,draw=gray!55},
        minor grid style={line width=0.15pt,draw=gray!35}
    }
}
\usepgfplotslibrary{fillbetween,groupplots}
}
        {}
    \begin{document}
    \def\FinalSimulationPlotEnd{\end{document}}
\fi
\IfFileExists{standalone_data_resolver.inc}
    {
\providecommand{\FinalSimulationResolveDataFile}[2]{%
    \IfFileExists{data_final/#1}{%
        \def#2{data_final/#1}%
    }{%
        \IfFileExists{../data_final/#1}{%
            \def#2{../data_final/#1}%
        }{%
            \PackageError{final-simulation-plot}{%
                Required data file #1 not found}{%
                Compile from JSAIT_2025_paper or final_simulation_plots, and
                ensure that JSAIT_2025_paper/data_final/#1 exists.}%
        }%
    }%
}
}
    {}

\ifdefined\AThreeCombinedCSV
\else
    \FinalSimulationResolveDataFile
        {A3_overcomplete_detection_beamforming_gain.csv}
        {\AThreeCombinedCSV}
\fi
\pgfplotstableread[col sep=comma]
    {\AThreeCombinedCSV}
    \AThreeCombinedData
\pgfplotstablegetelem{0}{T}\of\AThreeCombinedData
\edef\AThreeCombinedT{\pgfplotsretval}
\pgfplotstablegetelem{0}{Na}\of\AThreeCombinedData
\edef\AThreeCombinedNa{\pgfplotsretval}
\pgfplotstablegetelem{0}{Ng}\of\AThreeCombinedData
\edef\AThreeCombinedNg{\pgfplotsretval}
\pgfplotstablegetelem{0}{Nmc}\of\AThreeCombinedData
\edef\AThreeCombinedNmc{\pgfplotsretval}
\pgfplotstablegetelem{0}{epsilon}\of\AThreeCombinedData
\edef\AThreeCombinedEpsilon{\pgfplotsretval}
\pgfplotstablegetelem{0}{alpha_abs}\of\AThreeCombinedData
\edef\AThreeCombinedAlphaAbs{\pgfplotsretval}
\pgfplotstablegetelem{0}{DoA_min}\of\AThreeCombinedData
\edef\AThreeCombinedDoAMin{\pgfplotsretval}
\pgfplotstablegetelem{0}{DoA_max}\of\AThreeCombinedData
\edef\AThreeCombinedDoAMax{\pgfplotsretval}
\pgfplotstablegetcolsof{\AThreeCombinedData}
\edef\AThreeCombinedColumnCount{\pgfplotsretval}

\pgfplotsset{
    cp/.style={blue,line width=1pt,solid,mark=diamond*,mark size=1.7pt,
        mark repeat=4,mark phase=1,
        mark options={solid,fill=blue}},
    bose chowla/.style={red,line width=1pt,solid,mark=square*,mark size=1.7pt,
        mark repeat=4,mark phase=1,
        mark options={solid,fill=red}},
    singer/.style={green!55!black,line width=1pt,solid,mark=*,mark size=1.7pt,
        mark repeat=4,mark phase=1,
        mark options={solid,fill=green!55!black}},
    rm/.style={orange!85!black,line width=1pt,solid,mark=triangle*,mark size=1.8pt,
        mark repeat=4,mark phase=1,
        mark options={solid,fill=orange!85!black}},
    rvq/.style={black,line width=1pt,solid,mark=o,
        mark repeat=4,mark phase=2,
        mark size=2.4pt,
        mark options={solid,fill=white,draw=black,line width=0.85pt}},
    random phase/.style={violet!90!black,line width=1pt,dashed},
    random antenna selection/.style={cyan!80!blue,line width=1pt,
        mark repeat=4,mark phase=4,
        solid,mark=asterisk,mark size=3.0pt,
        mark options={solid,draw=cyan!80!blue,line width=1.0pt}}
}

\begin{figure}[t]
    \centering
    \begin{minipage}{\columnwidth}
        \centering
        \begin{tikzpicture}
        \begin{axis}[
            width=\columnwidth,height=4.5cm,
            xlabel={SNR (dB)},
            ylabel={$\PP(|\widehat f_{\rm md}-f^\star|> \epsilon)$
            },
            xmin=-10,xmax=10,ymin=2e-2,ymax=1,
            ymode=log,
            unbounded coords=jump,
            tick label style={font=\scriptsize},
            grid=both,
        major grid style={line width=0.25pt,draw=gray!55},
        minor grid style={line width=0.18pt,draw=gray!35},
            label style={font=\scriptsize},
            legend columns=1,
            legend style={
                at={(0.02,0.02)},anchor=south west,
                draw=white!60!black,
                fill=white,
                font=\scriptsize,
                /tikz/every even column/.append style={column sep=0.12cm}
            }
        ]
            \addplot[bose chowla] table[x=SNR,y=Plot_BC] {\AThreeCombinedData};
            \addlegendentry{Bose--Chowla}
            
            \addplot[rm] table[x=SNR,y=Plot_RM] {\AThreeCombinedData};
            \addlegendentry{RM (subspace)}
            
            \addplot[rvq]table[x=SNR,y=Plot_RVQ] {\AThreeCombinedData};
            \addlegendentry{RVQ}
            \addplot[random phase] table[x=SNR,y=Plot_RandomPhase] {\AThreeCombinedData};
            \addlegendentry{Random phase}
            
            \ifnum\AThreeCombinedColumnCount>94\relax
                \addplot[random antenna selection]
                    table[x=SNR,y=Plot_RandomAntennaSelection]
                    {\AThreeCombinedData};
                \addlegendentry{Random AS}
            \fi
        \end{axis}
        \end{tikzpicture}

        \par\smallskip
    \end{minipage}
    \caption{Off-grid detection error ($\epsilon\!=\!1/N_g$ for nearest gridpoint) when 
    true DoA is distrbuted as $f^\star\!\sim\!\mathrm{Uniform}([\AThreeCombinedDoAMin,\AThreeCombinedDoAMax))$ ($T\!=\!\AThreeCombinedT$, $N_a\!=\!N_g\!=\!\AThreeCombinedNa$). The performance of 
    CP and RM (subspace) BGCs deteriorates 
    off-grid unlike the 
    Bose--Chowla BGC.
    }
    \label{fig:final_line_packing_detection_error}
\end{figure}
\FinalSimulationPlotEnd

\ifdefined\FinalSimulationPlotsMaster
    \let\FinalSimulationPlotEnd\relax
    \def\AThreeDetectionFigureReference{%
        \cref{fig:final_line_packing_detection_error}}
\else
    \documentclass[class=IEEEtran,crop=false]{standalone}
    \IfFileExists{standalone_plot_preamble.inc}
        {
\IfFileExists{mathmacros.sty}{%
    \usepackage{mathmacros}%
}{%
    \IfFileExists{../mathmacros.sty}{%
        \input{../mathmacros.sty}%
    }{%
        \PackageError{standalone_plot_preamble}{mathmacros.sty not found}{%
            Compile from JSAIT_2025_paper or final_simulation_plots.}%
    }%
}
\usepackage{float}
\usepackage{xcolor}
\usepackage[capitalize]{cleveref}
\usepackage{tikz}
\usepackage{pgfplots}
\usepackage{pgfplotstable}
\usepackage[inline]{enumitem}
\pgfplotsset{
    compat=1.16,
    every axis/.append style={
        grid=both,
        minor tick num=1,
        every mark/.append style={solid},
        major grid style={line width=0.25pt,draw=gray!55},
        minor grid style={line width=0.15pt,draw=gray!35}
    }
}
\usepgfplotslibrary{fillbetween,groupplots}
}
        {}
    \begin{document}
    \def\FinalSimulationPlotEnd{\end{document}}
    \def\AThreeDetectionFigureReference{the corresponding detection-error
        experiment}
\fi
\IfFileExists{standalone_data_resolver.inc}
    {
\providecommand{\FinalSimulationResolveDataFile}[2]{%
    \IfFileExists{data_final/#1}{%
        \def#2{data_final/#1}%
    }{%
        \IfFileExists{../data_final/#1}{%
            \def#2{../data_final/#1}%
        }{%
            \PackageError{final-simulation-plot}{%
                Required data file #1 not found}{%
                Compile from JSAIT_2025_paper or final_simulation_plots, and
                ensure that JSAIT_2025_paper/data_final/#1 exists.}%
        }%
    }%
}
}
    {}

\ifdefined\AThreeCombinedCSV
\else
    \FinalSimulationResolveDataFile
        {A3_overcomplete_detection_beamforming_gain.csv}
        {\AThreeCombinedCSV}
\fi
\pgfplotstableread[col sep=comma]
    {\AThreeCombinedCSV}
    \AThreeCombinedData
\pgfplotstablegetcolsof{\AThreeCombinedData}
\edef\AThreeCombinedColumnCount{\pgfplotsretval}

\pgfplotstableread[col sep=comma]
    {\AThreeCombinedCSV}
    \AThreeCombinedData
\pgfplotstablegetelem{0}{T}\of\AThreeCombinedData
\edef\AThreeCombinedT{\pgfplotsretval}
\pgfplotstablegetelem{0}{Na}\of\AThreeCombinedData
\edef\AThreeCombinedNa{\pgfplotsretval}
\pgfplotstablegetelem{0}{Ng}\of\AThreeCombinedData
\edef\AThreeCombinedNg{\pgfplotsretval}
\pgfplotstablegetelem{0}{Nmc}\of\AThreeCombinedData
\edef\AThreeCombinedNmc{\pgfplotsretval}
\pgfplotstablegetelem{0}{epsilon}\of\AThreeCombinedData
\edef\AThreeCombinedEpsilon{\pgfplotsretval}
\pgfplotstablegetelem{0}{alpha_abs}\of\AThreeCombinedData
\edef\AThreeCombinedAlphaAbs{\pgfplotsretval}
\pgfplotstablegetelem{0}{DoA_min}\of\AThreeCombinedData
\edef\AThreeCombinedDoAMin{\pgfplotsretval}
\pgfplotstablegetelem{0}{DoA_max}\of\AThreeCombinedData
\edef\AThreeCombinedDoAMax{\pgfplotsretval}
\pgfplotstablegetcolsof{\AThreeCombinedData}
\edef\AThreeCombinedColumnCount{\pgfplotsretval}

\pgfplotsset{
    cp/.style={blue,line width=1pt,solid,mark=diamond*,mark size=1.7pt,
        mark repeat=4,mark phase=1,
        mark options={solid,fill=blue}},
    bose chowla/.style={red,line width=1pt,solid,mark=square*,mark size=1.7pt,
        mark repeat=4,mark phase=1,
        mark options={solid,fill=red}},
    singer/.style={green!55!black,line width=1pt,solid,mark=*,mark size=1.7pt,
        mark repeat=4,mark phase=1,
        mark options={solid,fill=green!55!black}},
    rm/.style={orange!85!black,line width=1pt,solid,mark=triangle*,mark size=1.8pt,
        mark repeat=4,mark phase=1,
        mark options={solid,fill=orange!85!black}},
    rvq/.style={black,line width=1pt,solid,mark=o,
        mark repeat=4,mark phase=2,
        mark size=2.4pt,
        mark options={solid,fill=white,draw=black,line width=0.85pt}},
    random phase/.style={violet!90!black,line width=1pt,dashed},
    random antenna selection/.style={cyan!80!blue,line width=1pt,
        mark repeat=4,mark phase=4,
        solid,mark=asterisk,mark size=3.0pt,
        mark options={solid,draw=cyan!80!blue,line width=1.0pt}}
}

\begin{figure}[t]
    \centering
    \begin{minipage}{\columnwidth}
        \centering
        \begin{tikzpicture}
        \begin{axis}[
            width=\columnwidth,height=4.5cm,
            xlabel={SNR (dB)},
            ylabel={$\overline{G}_{\rm BF}$},
            xmin=-10,xmax=10,ymin=0,ymax=0.8,
            tick label style={font=\scriptsize},
            label style={font=\scriptsize},
            title style={font=\small},
            grid=both,
        major grid style={line width=0.25pt,draw=gray!55},
        minor grid style={line width=0.18pt,draw=gray!35},
            legend columns=1,
            legend style={
                draw=white!60!black,
                at={(0.98,0.02)},anchor=south east,
                fill=white,
                font=\scriptsize,
                /tikz/every even column/.append style={column sep=0.12cm}
            }
        ]
            \addplot[bose chowla] table[x=SNR,y=BeamGain_BC] {\AThreeCombinedData};
            \addlegendentry{Bose--Chowla}
            
            \addplot[rm] table[x=SNR,y=BeamGain_RM] {\AThreeCombinedData};
            \addlegendentry{RM(subspace)}

            \addplot[rvq] table[x=SNR,y=BeamGain_RVQ] {\AThreeCombinedData};
            \addlegendentry{RVQ}
            
            
            \addplot[random phase] table[x=SNR,y=BeamGain_RandomPhase] {\AThreeCombinedData};
            \addlegendentry{Random phase}
            
            \addplot[random antenna selection] table[x=SNR,y=BeamGain_RandomAntennaSelection] {\AThreeCombinedData};
            \addlegendentry{Random AS}
            \addplot[gray!80!black,thick,dashed] table[x=SNR,y=GridConstrainedCSI_BeamGain] {\AThreeCombinedData};
            \draw[gray!80!black, thick, -stealth]
                (axis cs:-7,0.6) -- (axis cs:-7,0.78)
                node[at start,below] {\small Perfect CSI};
        \end{axis}
        \end{tikzpicture}

        \par\smallskip
    \end{minipage}
    \caption{Average downlink beamforming gain for off-grid true DoA $f^\star\sim \mathrm{Uniform}([\AThreeCombinedDoAMin,\AThreeCombinedDoAMax))$ ($T=\AThreeCombinedT$, $N_a=N_g=\AThreeCombinedNa$). The Bose--Chowla construction comes closest to the perfect CSI 
    upper bound while other methods saturate below this.
    }
    \label{fig:final_line_packing_beamforming_gain}
    \label{fig:final_line_packing_spectral_efficiency}
\end{figure}
\FinalSimulationPlotEnd

In practice, the true DoA need not lie on the designed grid. We continue with BGCs for 
$T=16,\ N_a=N_g=255$, 
simulating 
an off-grid setting where $f^\star\sim\mathrm{Uniform}([-1,1])$ and 
evaluating 
the probability that $\hat{f}_{\rm md}(\y)$ fails to identify the grid point nearest to $f^\star$, namely $\PP(|\hat{f}_{\rm md}-f^\star|>\epsilon)$ where $\epsilon=\frac{1}{N_g}$. 
\cref{fig:final_line_packing_detection_error} shows that, even in the off-grid setting, the Bose--Chowla construction continues to outperform the RM(subspace) BGC, as well as the randomized baselines.


We also study the effect of the improved DoA estimation performance of BGCs, used during the channel sensing phase, on beamforming gain, which is important in the communications phase. A larger beamgain can support, e.g., a higher achievable rate \cite{Love2019ssm,Dai2025}. Specifically, we consider the downlink normalized beamgain
$
    \overline G_{\rm BF}
    \triangleq \EE\!
    [
    |\h^H\w_{\rm c}|^2/
    \|\h\|_2^2
    ]
    =\EE\!
    [
    |\a^H(f^\star)\a(\hat f_{\rm md})|^2
    /
    N_a^2
    ]
    $
\cite{Love2019ssm}, corresponding to using beam $\w_{\rm c}=\a(\hat f_{\rm md})/\sqrt{N_a}$ on channel $\h=\alpha\a(f^\star)$. \cref{fig:final_line_packing_beamforming_gain} shows that the superior off-grid detection performance of the Bose--Chowla construction translates into the largest average beamforming gain among the considered methods.
As a benchmark, we also report the average beamforming gain when perfect CSI is available and the downlink beam direction is selected from 
$\cG$, so that 
$\overline G_{\rm BF, oracle}  \triangleq \EE\!\left[\max_{f_n\in\mathcal G} \frac{|\a^H(f^\star)\a(f_n)|^2}{N_a^2}\right].$
The Bose--Chowla construction approaches this benchmark and effectively attains it beyond 
$\approx 5$~dB SNR. 

\ifdefined\FinalSimulationPlotsMaster
    \let\FinalSimulationPlotEnd\relax
\else
    \documentclass[class=IEEEtran,crop=false]{standalone}
    \IfFileExists{standalone_plot_preamble.inc}
        {
\IfFileExists{mathmacros.sty}{%
    \usepackage{mathmacros}%
}{%
    \IfFileExists{../mathmacros.sty}{%
        \input{../mathmacros.sty}%
    }{%
        \PackageError{standalone_plot_preamble}{mathmacros.sty not found}{%
            Compile from JSAIT_2025_paper or final_simulation_plots.}%
    }%
}
\usepackage{float}
\usepackage{xcolor}
\usepackage[capitalize]{cleveref}
\usepackage{tikz}
\usepackage{pgfplots}
\usepackage{pgfplotstable}
\usepackage[inline]{enumitem}
\pgfplotsset{
    compat=1.16,
    every axis/.append style={
        grid=both,
        minor tick num=1,
        every mark/.append style={solid},
        major grid style={line width=0.25pt,draw=gray!55},
        minor grid style={line width=0.15pt,draw=gray!35}
    }
}
\usepgfplotslibrary{fillbetween,groupplots}
}
        {}
    \begin{document}
    \def\FinalSimulationPlotEnd{\end{document}}
\fi
\IfFileExists{standalone_data_resolver.inc}
    {
\providecommand{\FinalSimulationResolveDataFile}[2]{%
    \IfFileExists{data_final/#1}{%
        \def#2{data_final/#1}%
    }{%
        \IfFileExists{../data_final/#1}{%
            \def#2{../data_final/#1}%
        }{%
            \PackageError{final-simulation-plot}{%
                Required data file #1 not found}{%
                Compile from JSAIT_2025_paper or final_simulation_plots, and
                ensure that JSAIT_2025_paper/data_final/#1 exists.}%
        }%
    }%
}
}
    {}
\ifdefined\BOnePerformanceCSV
\else
    \FinalSimulationResolveDataFile{B1_roi_performance.csv}{\BOnePerformanceCSV}
\fi
\pgfplotstableread[col sep=comma]{\BOnePerformanceCSV}\BOneData
\pgfplotstablegetcolsof{\BOneData}
\edef\BOneColumnCount{\pgfplotsretval}
\pgfplotstablegetelem{0}{epsilon}\of\BOneData
\edef\BOneEpsilon{\pgfplotsretval}
\pgfplotstablegetelem{0}{T}\of\BOneData
\edef\BOneT{\pgfplotsretval}
\pgfplotstablegetelem{0}{Na}\of\BOneData
\edef\BOneNa{\pgfplotsretval}
\pgfplotstablegetelem{0}{Ng}\of\BOneData
\edef\BOneNg{\pgfplotsretval}
\pgfplotstablegetelem{0}{Nmc}\of\BOneData
\edef\BOneNmc{\pgfplotsretval}
\pgfplotstablegetelem{0}{ROI_min}\of\BOneData
\edef\BOneROIMin{\pgfplotsretval}
\pgfplotstablegetelem{0}{ROI_max}\of\BOneData
\edef\BOneROIMax{\pgfplotsretval}
\pgfplotstablegetelem{0}{alpha_abs}\of\BOneData
\edef\BOneAlphaAbs{\pgfplotsretval}
\pgfplotstablegetelem{0}{alpha_phase_rad}\of\BOneData
\edef\BOneAlphaPhase{\pgfplotsretval}
\pgfplotstablegetelem{0}{P_short}\of\BOneData
\edef\BOnePShort{\pgfplotsretval}
\pgfplotstablegetelem{0}{P_long}\of\BOneData
\edef\BOnePLong{\pgfplotsretval}

\ifdefined\BOnePerformanceCSV
\else
    \FinalSimulationResolveDataFile{B1_roi_performance.csv}{\BOnePerformanceCSV}
\fi
\ifdefined\BOneBeamgainCSV
\else
    \FinalSimulationResolveDataFile{B1_total_beamgain.csv}{\BOneBeamgainCSV}
\fi
\pgfplotstableread[col sep=comma]{\BOnePerformanceCSV}\BOneData
\pgfplotstableread[col sep=comma]{\BOneBeamgainCSV}\BOneGainData
\pgfplotstablegetcolsof{\BOneData}
\edef\BOneColumnCount{\pgfplotsretval}
\pgfplotstablegetelem{0}{T}\of\BOneData
\edef\BOneT{\pgfplotsretval}
\pgfplotstablegetelem{0}{Na}\of\BOneData
\edef\BOneNa{\pgfplotsretval}
\pgfplotstablegetelem{0}{Ng}\of\BOneData
\edef\BOneNg{\pgfplotsretval}
\pgfplotstablegetelem{0}{ROI_min}\of\BOneData
\edef\BOneROIMin{\pgfplotsretval}
\pgfplotstablegetelem{0}{ROI_max}\of\BOneData
\edef\BOneROIMax{\pgfplotsretval}
\pgfplotstablegetelem{0}{alpha_abs}\of\BOneData
\edef\BOneAlphaAbs{\pgfplotsretval}
\pgfplotstablegetelem{0}{P_short}\of\BOneData
\edef\BOnePShort{\pgfplotsretval}
\pgfplotstablegetelem{0}{P_long}\of\BOneData
\edef\BOnePLong{\pgfplotsretval}

\pgfplotsset{
    cp/.style={blue,line width=1pt,solid,
    mark size=1.7pt,
        mark repeat=4,mark phase=1,
        mark options={solid,fill=blue}},
    bose chowla/.style={red,line width=1pt,solid,mark=square*,mark size=1.7pt,
        mark repeat=4,mark phase=1,
        mark options={solid,fill=red}},
    singer/.style={green!55!black,line width=1pt,solid,mark=*,mark size=1.7pt,
        mark repeat=4,mark phase=1,
        mark options={solid,fill=green!55!black}},
    rm/.style={orange!85!black,line width=1pt,solid,mark=triangle*,mark size=1.8pt,
        mark repeat=4,mark phase=1,
        mark options={solid,fill=orange!85!black}},
    rvq/.style={black,line width=1pt,solid,mark=o,
        mark repeat=4,mark phase=2,
        mark size=2.4pt,
        mark options={solid,fill=white,draw=black,line width=0.85pt}},
    random phase/.style={violet!90!black,line width=1pt,dashed},
    random antenna selection/.style={cyan!80!blue,line width=1pt,
        mark repeat=4,mark phase=4,
        solid,mark=asterisk,mark size=3.0pt,
        mark options={solid,draw=cyan!80!blue,line width=1.0pt}}
}

\begin{figure}[h]
    \centering
    \begin{minipage}{\columnwidth}
    \centering
    \begin{tikzpicture}
    \begin{axis}[
        width=\columnwidth,height=4cm,
        xlabel={Spatial frequency $f$},ylabel={$\|\W\a(f)\|_2^2$},
        xmin=-1,xmax=1,ymin=0,ymax=70,
        title={(a) Total beamforming gain},
        title style={font=\small},
        tick label style={font=\scriptsize},
        label style={font=\scriptsize},
        set layers=axis on top,
        legend to name=CBSBeamGainLegend,
        legend columns=4,
        legend style={
                draw=white!60!black,
                fill=white,
                font=\scriptsize,
                /tikz/every even column/.append style={column sep=0.12cm}
            }
    ]
        \begin{pgfonlayer}{axis background}
            \path[fill=blue!8,draw=none]
                (axis cs:\BOneROIMin,0) rectangle
                (axis cs:\BOneROIMax,70);
        \end{pgfonlayer}
        \node[font=\scriptsize,anchor=north,fill=white,fill opacity=0.75,
            text opacity=1,inner sep=1.5pt]
            at (axis cs:0,58) {ROI};
        
        \addplot[cp, style={dash dot, mark=none}] table[x=f,y=Gain_CBS_P3] {\BOneGainData};
        \addlegendentry{CBS, $P=\BOnePShort$}
        
        \addplot[cp, style={mark=none}] table[x=f,y=Gain_CBS_P56] {\BOneGainData};
        \addlegendentry{CBS, $P=\BOnePLong$}
        
        \addplot[bose chowla, style={mark repeat=50, mark phase = 1}] table[x=f,y=Gain_BC] {\BOneGainData};
        \addlegendentry{Bose--Chowla}
        

        \addplot[rvq, style={mark repeat=50, mark phase = 30}] table[x=f,y=Gain_RVQ] {\BOneGainData};
        \addlegendentry{RVQ}

        \addplot[rm, style={mark repeat=50, mark phase = 15}] table[x=f,y=Gain_RM] {\BOneGainData};
        \addlegendentry{RM (subspace)}
        
        
        \addplot[random phase, style={solid, line width=0.5pt}] table[x=f,y=Gain_RandomPhase] {\BOneGainData};
        \addlegendentry{Random phase}
        
        \addplot[random antenna selection, style={mark repeat=50, mark phase = 45}] table[x=f,y=Gain_RandomAntennaSelection] {\BOneGainData};
        \addlegendentry{Random AS}
        
        \addplot[black,dashed]
            coordinates {(\BOneROIMin,0) (\BOneROIMin,70)};
        \addplot[black,dashed]
            coordinates {(\BOneROIMax,0) (\BOneROIMax,70)};
    \end{axis}
    \end{tikzpicture}
\\
    \begin{tikzpicture}
    \begin{axis}[
        width=0.98\columnwidth,height=4cm,
        xlabel={SNR (dB)},ylabel={$\PP(|\widehat f_{\rm md}-f^\star|>\epsilon)$
        },
        xmin=-10,xmax=10,ymin=1e-2,ymax=1,
        grid=both,
        major grid style={line width=0.25pt,draw=gray!55},
        minor grid style={line width=0.18pt,draw=gray!35},
        ymode=log,unbounded coords=jump,
        title={(b) Detection error vs. SNR},
        title style={font=\small},
        tick label style={font=\scriptsize},
            label style={font=\scriptsize},
        legend to name=CBSResultsLegend,
         legend columns=4,
            legend style={
                draw=white!60!black,
                at={(0.02,0.02)},anchor=south west,
                fill=white,
                font=\scriptsize,
                /tikz/every even column/.append style={column sep=0.12cm}
            }
    ]
        \addplot[cp, style={dash dot}] table[x=SNR,y=Plot_CBS_P3] {\BOneData};
        \addlegendentry{CBS, $P=\BOnePShort$}
        
        \addplot[cp, style={mark phase=2}] table[x=SNR,y=Plot_CBS_P56] {\BOneData};
        \addlegendentry{CBS, $P=\BOnePLong$}
        
        \addplot[bose chowla] table[x=SNR,y=Plot_BC] {\BOneData};
        \addlegendentry{Bose--Chowla}
        

        \addplot[rvq] table[x=SNR,y=Plot_RVQ] {\BOneData};
        \addlegendentry{RVQ}

        \addplot[rm]table[x=SNR,y=Plot_RM] {\BOneData};
        \addlegendentry{RM(subspace)}
        
        
        \addplot[random phase] table[x=SNR,y=Plot_RandomPhase] {\BOneData};
        \addlegendentry{Random phase}
        
        \addplot[random antenna selection] table[x=SNR,y=Plot_RandomAntennaSelection] {\BOneData};
        \addlegendentry{Random AS}
    \end{axis}
    \end{tikzpicture}
    \par\smallskip
    \pgfplotslegendfromname{CBSResultsLegend}
    \end{minipage}
    \caption{ Off-grid DoA $f^\star \sim \mathrm{Uniform}([\BOneROIMin,\BOneROIMax])$ with $T=\BOneT$, $N_a=N_g=\BOneNg$. (a)Total beamforming gain for spatial frequency on grid $\cG$ for the ROI $[\BOneROIMin,\BOneROIMax]$. CBS 
    allows a high, controlled, beamgain within the ROI.  (b) Detection error probability ($\epsilon=1/N_g$) vs SNR. 
    The larger total beamforming gain obtained by CBS leads to a better performance in detection error probability.
    }
    \label{fig:cbs_results}
\end{figure}
\FinalSimulationPlotEnd

\subsection{Spatially Selective Case with Prior DoA Knowledge}\label{sec:sim_spatially_selective} 

Next, we consider the same setting ($T=\BOneT$, $N_a=\BOneNa$, $N_g=\BOneNg$) with the exception that the off-grid true DoA $f^\star$ is known to lie within a region of interest (ROI) $[\BOneROIMin,\BOneROIMax]$. Hence, we can use this prior knowledge to improve sensing performance by beamforming in the ROI. Convolutional Beamspace (CBS) enables flexibly choosing the beamformer without affecting the minimum subspace distance of the induced BGCs. We consider CBS with shifts given by the Bose--Chowla construction with minimum aperture. 

\cref{fig:cbs_results}(a) shows the total beamforming gain for CBS with a length $P=\BOnePShort$ sinc filter $\w = [1,1,1]$, and $P=\BOnePLong$ filter designed using the Parks--McClellan algorithm, along with total beamgain of Bose--Chowla and RM(subspace) BGCs. For comparison, we also show the total beamgain of a realization 
of the three 
random  beamformer baselines. 
By design,  CBS allows flexibility in focusing the beamgain within the ROI while maintaining the minimum subspace distance of the 
Bose--Chowla construction.
\cref{fig:cbs_results}(b) shows the improvement in off-grid probability of detection error (i.e. $\PP(|\hat{f}_{\rm md}-f^\star|>\epsilon)$ where $\epsilon=1/N_g$) performance due to the increase in beamforming gain over the ROI. Specifically, the difference in sensing performance of Bose--Chowla CBS with filter lengths $P=\BOnePShort$ and $P=\BOnePLong$ is due to the different beamforming gains of the two cases, since both employ the same Bose--Chowla shifts and thereby have the same $d^{(s)}_{\min}$.



\subsection{Non-Left-Invertible Array Manifolds}
\ifdefined\FinalSimulationPlotsMaster
    \let\FinalSimulationPlotEnd\relax
\else
    \documentclass[class=IEEEtran,crop=false]{standalone}
    \IfFileExists{standalone_plot_preamble.inc}
        {
\IfFileExists{mathmacros.sty}{%
    \usepackage{mathmacros}%
}{%
    \IfFileExists{../mathmacros.sty}{%
        \input{../mathmacros.sty}%
    }{%
        \PackageError{standalone_plot_preamble}{mathmacros.sty not found}{%
            Compile from JSAIT_2025_paper or final_simulation_plots.}%
    }%
}
\usepackage{float}
\usepackage{xcolor}
\usepackage[capitalize]{cleveref}
\usepackage{tikz}
\usepackage{pgfplots}
\usepackage{pgfplotstable}
\usepackage[inline]{enumitem}
\pgfplotsset{
    compat=1.16,
    every axis/.append style={
        grid=both,
        minor tick num=1,
        every mark/.append style={solid},
        major grid style={line width=0.25pt,draw=gray!55},
        minor grid style={line width=0.15pt,draw=gray!35}
    }
}
\usepgfplotslibrary{fillbetween,groupplots}
}
        {}
    \begin{document}
    \def\FinalSimulationPlotEnd{\end{document}}
\fi

\def\AFiveNgTSeventeenResolveDataPath#1{%
    \IfFileExists{data_final/#1}{%
        \def\AFiveNgTSeventeenData{data_final/#1}%
    }{%
        \IfFileExists{../data_final/#1}{%
            \def\AFiveNgTSeventeenData{../data_final/#1}%
        }{%
            \def\AFiveNgTSeventeenData{}%
        }%
    }%
}
\AFiveNgTSeventeenResolveDataPath{fig_A5_linePacking_NgGreaterNa_diagnostic_T17.csv}
\ifx\AFiveNgTSeventeenData\empty
    \PackageError{fig_A5_linePacking_NgGreaterNa_diagnostic_T17}{
        Required T=17 simulation CSV was not found}{
        Compile from JSAIT_2025_paper or final_simulation_plots.}
\fi

\pgfplotstableread[col sep=comma]
    {\AFiveNgTSeventeenData}\AFiveNgTSeventeenTable
\pgfplotstablegetelem{0}{T}\of\AFiveNgTSeventeenTable
\edef\AFiveNgTSeventeenT{\pgfplotsretval}
\pgfplotstablegetelem{0}{Nmc}\of\AFiveNgTSeventeenTable
\edef\AFiveNgTSeventeenNmc{\pgfplotsretval}
\pgfplotstablegetelem{0}{N_random_codebooks}\of\AFiveNgTSeventeenTable
\edef\AFiveNgTSeventeenNRandom{\pgfplotsretval}
\pgfplotstablegetelem{0}{Ng_A_Na288_Ng288}\of\AFiveNgTSeventeenTable
\edef\AFiveNgTSeventeenNaturalNg{\pgfplotsretval}
\pgfplotstablegetelem{0}{Na_A_Na226_Ng288}\of\AFiveNgTSeventeenTable
\edef\AFiveNgTSeventeenBCNa{\pgfplotsretval}
\pgfplotstablegetelem{0}{Na_A_Na144_Ng288}\of\AFiveNgTSeventeenTable
\edef\AFiveNgTSeventeenHalfNa{\pgfplotsretval}
\pgfplotstablegetelem{0}{Ng_B_Na288_Ng576}\of\AFiveNgTSeventeenTable
\edef\AFiveNgTSeventeenDoubleNg{\pgfplotsretval}
\pgfplotstablegetelem{0}{Ng_B_Na288_Ng864}\of\AFiveNgTSeventeenTable
\edef\AFiveNgTSeventeenTripleNg{\pgfplotsretval}
\pgfplotstablegetelem{0}{bc_minimum_aperture_span}\of\AFiveNgTSeventeenTable
\edef\AFiveNgTSeventeenBCSpan{\pgfplotsretval}
\pgfplotstablegetelem{0}{error_epsilon_panel_C}\of\AFiveNgTSeventeenTable
\edef\AFiveNgTSeventeenDetectionEpsilon{\pgfplotsretval}

\pgfplotsset{
    cp/.style={blue,line width=1pt,solid,mark=diamond*,mark size=1.7pt,
        mark repeat=4,mark phase=1,
        mark options={solid,fill=blue}},
    bose chowla/.style={red,line width=1pt,solid,mark=square*,mark size=1.7pt,
        mark repeat=4,mark phase=1,
        mark options={solid,fill=red}},
    singer/.style={green!55!black,line width=1pt,solid,mark=*,mark size=1.7pt,
        mark repeat=4,mark phase=1,
        mark options={solid,fill=green!55!black}},
    rm/.style={orange!85!black,line width=1pt,solid,mark=triangle*,mark size=1.8pt,
        mark repeat=4,mark phase=1,
        mark options={solid,fill=orange!85!black}},
    rvq/.style={black,line width=1pt,solid,mark=o,
        mark repeat=4,mark phase=2,
        mark size=2.4pt,
        mark options={solid,fill=white,draw=black,line width=0.85pt}},
    random phase/.style={violet!90!black,line width=1pt,dashed},
    random antenna selection/.style={cyan!80!blue,line width=1pt,
        mark repeat=4,mark phase=4,
        solid,mark=asterisk,mark size=3.0pt,
        mark options={solid,draw=cyan!80!blue,line width=1.0pt}}
}

\begin{figure}[t]
\centering

\begin{tikzpicture}
\begin{axis}[
    width=\columnwidth,
    height=4.5cm,
    xlabel={SNR (dB)},
    ylabel={$\PP(|\widehat f_{\rm md}-f^\star|>\epsilon)$},
    xmin=-5,xmax=10,
    ymin=1e-4,ymax=1,
    ymode=log,
    ytick={1e-4,1e-3,1e-2,1e-1,1},
    minor y tick num=9,
    grid=both,
    major grid style={line width=0.25pt,draw=gray!55},
    minor grid style={line width=0.18pt,draw=gray!35},
    unbounded coords=jump,
    tick label style={font=\scriptsize},
    label style={font=\scriptsize},
    title style={font=\small},
    legend columns=4,
    legend style={
        at={(0.5,-0.34)},anchor=north,
        draw=white!60!black,fill=white,
        font=\tiny,cells={anchor=west},
        legend image post style={xscale=0.75},
        /tikz/every even column/.append style={column sep=0.04cm}
    }
]
     \addlegendimage{empty legend}
    \addlegendentry{$N_g=\AFiveNgTSeventeenNaturalNg$:}
    
    

 \addplot[bose chowla, mark phase=1]table[x=SNR,y=Plot_BC_C_Na288_Ng288] {\AFiveNgTSeventeenTable};
    \addlegendentry{BC}

    \addplot[cp, mark phase=2]table[x=SNR,y=Plot_CP_C_Na288_Ng288]{\AFiveNgTSeventeenTable};
    \addlegendentry{CP}
     
     \addplot[rvq, mark phase=3] table[x=SNR,y=Plot_RVQ_C_Na288_Ng288]{\AFiveNgTSeventeenTable};
    \addlegendentry{RVQ}

    \addlegendimage{empty legend}
    \addlegendentry{$N_g=\AFiveNgTSeventeenDoubleNg$:}

    \addplot[bose chowla, style={dashed, mark phase=1}] table[x=SNR,y=Plot_BC_C_Na288_Ng576] {\AFiveNgTSeventeenTable};
    \addlegendentry{BC}
    
    \addplot[cp, style={dashed, mark phase=2}] table[x=SNR,y=Plot_CP_C_Na288_Ng576] {\AFiveNgTSeventeenTable};
     \addlegendentry{CP}
   
    \addplot[rvq, style={dashed, mark phase=3}] table[x=SNR,y=Plot_RVQ_C_Na288_Ng576] {\AFiveNgTSeventeenTable};
    \addlegendentry{RVQ}

    \addlegendimage{empty legend}
    \addlegendentry{$N_g=\AFiveNgTSeventeenTripleNg$}

        \addplot[bose chowla, style={dotted, mark phase=1}] table[x=SNR,y=Plot_BC_C_Na288_Ng864] {\AFiveNgTSeventeenTable};
        \addlegendentry{BC}

        \addplot[cp, style={dotted, mark phase=2}] table[x=SNR,y=Plot_CP_C_Na288_Ng864] {\AFiveNgTSeventeenTable};
        \addlegendentry{CP}
        
        \addplot[rvq, style={dotted, mark phase=3}] table[x=SNR,y=Plot_RVQ_C_Na288_Ng864] {\AFiveNgTSeventeenTable};
        \addlegendentry{RVQ}
    
\end{axis}
\end{tikzpicture}

\caption{Effect of increasing 
grid size $N_g$ for $T=\AFiveNgTSeventeenT$ and fixed $N_a=T^2-1=\AFiveNgTSeventeenNaturalNg$ 
. True DoA $f^\star$ is chosen uniformly at random from grid $\cG$. 
Even for $N_g>N_a$, the estimates of Bose--Chowla construction are located within an  $\epsilon=1/N_a$ neighborhood of  $f^\star$ 
with high probability. 
}
\label{fig:line_packing_Ng_greater_Na_diagnostic}
\end{figure}
\FinalSimulationPlotEnd

\ifdefined\FinalSimulationPlotsMaster
    \let\FinalSimulationPlotEnd\relax
\else
    \documentclass[class=IEEEtran,crop=false]{standalone}
    \IfFileExists{standalone_plot_preamble.inc}
        {
\IfFileExists{mathmacros.sty}{%
    \usepackage{mathmacros}%
}{%
    \IfFileExists{../mathmacros.sty}{%
        \input{../mathmacros.sty}%
    }{%
        \PackageError{standalone_plot_preamble}{mathmacros.sty not found}{%
            Compile from JSAIT_2025_paper or final_simulation_plots.}%
    }%
}
\usepackage{float}
\usepackage{xcolor}
\usepackage[capitalize]{cleveref}
\usepackage{tikz}
\usepackage{pgfplots}
\usepackage{pgfplotstable}
\usepackage[inline]{enumitem}
\pgfplotsset{
    compat=1.16,
    every axis/.append style={
        grid=both,
        minor tick num=1,
        every mark/.append style={solid},
        major grid style={line width=0.25pt,draw=gray!55},
        minor grid style={line width=0.15pt,draw=gray!35}
    }
}
\usepgfplotslibrary{fillbetween,groupplots}
}
        {}
    \begin{document}
    \def\FinalSimulationPlotEnd{\end{document}}
\fi
\IfFileExists{standalone_data_resolver.inc}
    {
\providecommand{\FinalSimulationResolveDataFile}[2]{%
    \IfFileExists{data_final/#1}{%
        \def#2{data_final/#1}%
    }{%
        \IfFileExists{../data_final/#1}{%
            \def#2{../data_final/#1}%
        }{%
            \PackageError{final-simulation-plot}{%
                Required data file #1 not found}{%
                Compile from JSAIT_2025_paper or final_simulation_plots, and
                ensure that JSAIT_2025_paper/data_final/#1 exists.}%
        }%
    }%
}
}
    {}
\FinalSimulationResolveDataFile{fig_C1_alpha_estimation_with_CP_settings.csv}
    {\COneCPSettingsCSV}
\FinalSimulationResolveDataFile{fig_C1_alpha_estimation_with_CP_panel_a.csv}
    {\COneCPPanelACSV}
\FinalSimulationResolveDataFile{fig_C1_alpha_estimation_with_CP_panel_b.csv}
    {\COneCPPanelBCSV}
\pgfplotstableread[col sep=comma]{\COneCPSettingsCSV}\COneCPSettings
\pgfplotstablegetelem{0}{T_doa}\of\COneCPSettings
\edef\COneCPTDoA{\pgfplotsretval}
\pgfplotstablegetelem{0}{Na}\of\COneCPSettings
\edef\COneCPNa{\pgfplotsretval}
\pgfplotstablegetelem{0}{Ng}\of\COneCPSettings
\edef\COneCPNg{\pgfplotsretval}
\pgfplotstablegetelem{0}{Nmc}\of\COneCPSettings
\edef\COneCPNmc{\pgfplotsretval}
\pgfplotstablegetelem{0}{T_est_1}\of\COneCPSettings
\edef\COneCPTEstOne{\pgfplotsretval}
\pgfplotstablegetelem{0}{T_total_1}\of\COneCPSettings
\edef\COneCPTotalOne{\pgfplotsretval}
\pgfplotstablegetelem{0}{T_est_2}\of\COneCPSettings
\edef\COneCPTEstTwo{\pgfplotsretval}
\pgfplotstablegetelem{0}{T_total_2}\of\COneCPSettings
\edef\COneCPTotalTwo{\pgfplotsretval}
\pgfplotstablegetelem{0}{FixedAlphaMagnitude}\of\COneCPSettings
\edef\COneCPAlphaAbs{\pgfplotsretval}
\pgfplotstablegetelem{0}{GaussianAlphaVariance}\of\COneCPSettings
\edef\COneCPAlphaVar{\pgfplotsretval}

\pgfplotsset{
    cOneCPPlot/.style={line width=1.2pt},
    cOneCPLongDash/.style={dash pattern=on 7pt off 3pt}
}

\begin{figure*}[t]
\centering
\begin{tikzpicture}
\begin{axis}[
    width=\columnwidth,height=4cm,
    title={(a) Deterministic unknown setting},
    xlabel={SNR (dB)},
    ylabel={$\PP(\hat{f}_{\rm md} \neq f^\star)$},
    xmin=-10,xmax=10,ymin=8e-5,ymax=1,
    ymode=log,unbounded coords=jump,
    ytick={1e-4,1e-3,1e-2,1e-1,1},
    minor y tick num=9,
    grid=both,
        major grid style={line width=0.25pt,draw=gray!55},
        minor grid style={line width=0.18pt,draw=gray!35},
    tick label style={font=\scriptsize},
    label style={font=\scriptsize},title style={font=\small},
    legend to name=COneCPAlphaLegend,
    legend columns=4,
    legend style={draw=white!60!black,fill=white,font=\tiny,
        /tikz/every even column/.append style={column sep=0.08cm}}
]
    
    \addlegendimage{empty legend}
    \addlegendentry{RM(channel code):}

    \addplot[black,very thick,solid,mark=star,mark size=2.5pt,
        mark options={solid}]
        table[col sep=comma,x=SNR_dB,
            y expr={max(\thisrow{Pe_RM_Tdoa16_known_alpha},1e-6)}]
        {\COneCPPanelACSV};
    \addlegendentry{
    $T_{\text{doa}}=\COneCPTDoA,T_{\text{est}}=0$
    }
    
    \addplot[blue,cOneCPPlot,dotted]
        table[col sep=comma,x=SNR_dB,
            y expr={max(\thisrow{Pe_RM_Tdoa16_Test1},1e-6)}]
        {\COneCPPanelACSV};
    \addlegendentry{$T_{\text{doa}}=\COneCPTDoA,T_{\text{est}}=\COneCPTEstOne$}

    \addplot[blue,cOneCPPlot,cOneCPLongDash]
        table[col sep=comma,x=SNR_dB,
            y expr={max(\thisrow{Pe_RM_Tdoa16_Test7},1e-6)}]
        {\COneCPPanelACSV};
    \addlegendentry{
    $T_{\text{doa}}=\COneCPTDoA,T_{\text{est}}=\COneCPTEstTwo$
    }

    \addlegendimage{empty legend}
    \addlegendentry{CP (subspace code):}

    \addplot[red,cOneCPPlot,solid,mark=star,mark size=2.5pt,
        mark options={solid}]
        table[col sep=comma,x=SNR_dB,
            y expr={max(\thisrow{Pe_CP_T16},1e-6)}]
        {\COneCPPanelACSV};
    \addlegendentry{$T_{\text{total}}=16$}
    
    \addplot[red,cOneCPPlot,dotted]
        table[col sep=comma,x=SNR_dB,
            y expr={max(\thisrow{Pe_CP_T17},1e-6)}]
        {\COneCPPanelACSV};
    \addlegendentry{
    $T_{\text{total}}=\COneCPTotalOne$}
    
    \addplot[red,cOneCPPlot,cOneCPLongDash]
        table[col sep=comma,x=SNR_dB,
            y expr={max(\thisrow{Pe_CP_T23},1e-6)}]
        {\COneCPPanelACSV};
    \addlegendentry{$T_{\text{total}}=\COneCPTotalTwo$}

\end{axis}
\end{tikzpicture}
\begin{tikzpicture}
\begin{axis}[
    width=\columnwidth,height=4cm,
    title={(b) Stochastic unknown setting},
    xlabel={SNR (dB)},
    ylabel={$\PP(\hat{f}_{\rm md} \neq f^\star)$},
    xmin=-10,xmax=10,ymin=5e-2,ymax=1,
    ymode=log,unbounded coords=jump,
    ytick={1e-2,1e-1,1},
    minor y tick num=9,
    grid=both,
        major grid style={line width=0.25pt,draw=gray!55},
        minor grid style={line width=0.18pt,draw=gray!35},
    tick label style={font=\scriptsize},
    label style={font=\scriptsize},title style={font=\small}
]
   
     \addplot[blue,cOneCPPlot,cOneCPLongDash]
        table[col sep=comma,x=SNR_dB,y=Pe_RM_Tdoa16_Test7]
        {\COneCPPanelBCSV};
      \addplot[blue,cOneCPPlot,dotted]
        table[col sep=comma,x=SNR_dB,y=Pe_RM_Tdoa16_Test1]
        {\COneCPPanelBCSV};
        
    \addplot[red,cOneCPPlot,solid,mark=star,mark size=2.5pt,
        mark options={solid}]
        table[col sep=comma,x=SNR_dB,y=Pe_CP_T16]
        {\COneCPPanelBCSV};
    \addplot[red,cOneCPPlot,dotted]
        table[col sep=comma,x=SNR_dB,y=Pe_CP_T17]
        {\COneCPPanelBCSV};

    \addplot[red,cOneCPPlot,cOneCPLongDash]
        table[col sep=comma,x=SNR_dB,y=Pe_CP_T23]
        {\COneCPPanelBCSV};

     \addplot[black,very thick,solid,mark=star,mark size=2.5pt,
        mark options={solid}]
        table[col sep=comma,x=SNR_dB,y=Pe_RM_Tdoa16_known_alpha]
        {\COneCPPanelBCSV};
    
\end{axis}
\end{tikzpicture}
\par\smallskip
\ref{COneCPAlphaLegend}
\caption{
Comparison of channel- 
and 
subspace-coding approaches for on-grid DoA sensing 
($N_a\!=\!N_g\!=\!2048$). 
RM-based channel coding uses $T_{\rm est}$ measurements to estimate $\alpha$ and $T_{\rm doa}$ for DoA sensing. CP-based subspace coding uses all $T_{\rm total}\!=\!T_{\rm est}\!+\!T_{\rm doa}$ measurements for DoA sensing. 
Allocating all measurements to DoA sensing 
reduces 
error probability for both (a) deterministic and (b) stochastic $(\alpha,f)$.
}
\label{fig:t_est_sweep_with_cp}
\end{figure*}
\FinalSimulationPlotEnd

When $N_g>N_a$, the array manifold is non-left-invertible, so an arbitrary target response $\B_G$ is generally not 
realizable. However, as established in \cref{lem:realizability_NgGreaterThanNa},  
Bose--Chowla and Singer constructions  take the form $\B_{\sf X}=\S_{\Omega_{\sf X}}\A_{\UU}$ and therefore remain realizable in this regime by antenna selection. The same row selections extend to denser grids using the corresponding $N_g\!\times\!N_g$ DFT matrix (cf. \cref{sec:BGC_Ng_greaterthan_Na}), although their minimum distance 
may no longer be near-optimal.


\Cref{fig:line_packing_Ng_greater_Na_diagnostic} studies their performance for \(T=17\) and fixed \(N_a=288 =T^2-1\), with \(N_g\in\{288,576,864\}\).
Similar to 
{\cite{Javidi2019ie8n},}
we consider an $\epsilon$-error probability 
given by \(\PP(|\widehat f_{\rm md}-f^\star|>\epsilon)\), with a tolerance level of \(\epsilon=1/N_a\). 
This corresponds to the natural resolution (inversely proportional to the aperture $N_a$) of the underlying 
ULA, independent of the grid size (and applicable even in a gridless setting). It is striking that the (extended) Bose--Chowla construction retains favorable performance under this criterion even when \(N_g>N_a\).  
This implies that the erroneous DoA estimates produced by Bose--Chowla remain clustered within an $\epsilon = \frac{1}{N_a}$ neighborhood around the true DoA as the grid size grows, instead of being distributed farther away from it. This 
suggests 
an interesting direction for future 
work 
on 
$\epsilon-$recovery performance with increasing grid size, 
in the spirit of rate-distortion theory. 


In contrast, the CP and RVQ target response is not guaranteed to satisfy row-space condition \eqref{eqn:SuffProj} when \(N_g>N_a\), so \(\W_{\rm CP}=\B_{\rm CP}\A_{\UU}^{\dagger}\) and \(\W_{\rm RVQ}=\B_{\rm RVQ}\A_{\UU}^{\dagger}\) generally realize only the projection of corresponding target responses onto the row space of \(\A_{\UU}\). Case \(N_g=288\) uses a subset of \(\mathrm{CP}(2,17)\) with codebook size \(17^2=289\) and \(d_{\min}^{(s)}\geq1 - \frac{1}{17}\) (similar to Bose--Chowla), while the denser cases require subset selection from higher order CP codes $\mathrm{CP}(3,17)$ (with \(17^3=4913\) codewords) so that a minimum distance of \(d_{\min}^{(s)}\geq1 - \frac{4}{17}\) is guaranteed \cite{Mahdavifar2022}. 
The combination of projection distortion and the reduced minimum distance with increasing $N_g$ contribute to the sharper degradation of CP.

\subsection{Should Channel Coefficient $\alpha$ be Estimated?}
\label{subsec:alpha_estimation}

\cref{sec:relation_ML_MAP_minimumdist} showed that the MD estimator of DoA $f$ for unknown $\alpha$ (which is also the ML estimator) does \emph{not} require estimating $\alpha$. However, one may ask whether it is still beneficial to first estimate $\alpha$ and then use it for estimating $f$. 
We compare two ways of handling the unknown channel coefficient. 
``\emph{CP(subspace)}'' 
employs a CP-code-based BGC
and applies the MD 
DoA decoder as given in \eqref{eq:min_dist_subspace_decoder}, 
thereby avoiding explicit estimation of $\alpha$. On the other hand, the coherent channel-coding method of \cite{Yu2025ongrid} utilizes a BPSK-mapped $\mathrm{RM}$ codebook 
and applies the minimum-$\ell_2$-distance decoder 
\eqref{eq:ML_estimator_known_alpha}. We refer to this method as ``\emph{RM(known $\alpha$)}" when the decoder knows $\alpha$ exactly, and as ``\emph{RM(unknown $\alpha$)}" when $\alpha$ is first estimated using additional $T_{\rm est}$ measurements and then substituted into the coherent decoder \eqref{eq:ML_estimator_known_alpha}. Following \cite{Yu2025ongrid}, we use a beamformer with a constant beampattern given by $\widetilde{\w} = (\A^{\dagger}_{\UU})^H\mathbf{1}_{N_g\times 1}$ to collect measurements $\tilde{y}_t = \alpha + z_t,\ t=1,\dots,T_{\rm est}$ where $z_t = \widetilde{\w}^H\n_t$ for estimating $\alpha$.



We consider an on-grid true DoA with $N_a=N_g=2048$ to ensure exact realizability of both RM and CP code-based beamformers.
\footnote{Note that in this $(T,N_g)$ regime, optimal Bose Chowla or Singer codes do not exist, but CP codes with good subspace distance can be designed \cite{Mahdavifar2022}.} 
All RM variants use $T_{\rm doa}=16$ coded measurements for DoA estimation. \emph{RM(known $\alpha$)} requires no additional measurements (i.e. $T_{\rm est}=0$), whereas \emph{RM(unknown $\alpha$)} uses $T_{\rm est}\in\{1,7\}$ additional measurements to estimate $\alpha$, giving total sensing budgets $T_{\rm total}=T_{\rm doa} + T_{\rm est}$ of $17$ and $23$, respectively. For a fair comparison, \emph{CP(subspace)} uses the same total budgets $T_{\rm total} \in\{16,17,23\}$, but devotes every measurement to subspace-coded sensing.

In \cref{fig:t_est_sweep_with_cp}(a), the DoA and $\alpha$ are deterministic unknowns. The true DoA is set to the gridpoint closest to $f=0$. 
\emph{RM(unknown $\alpha$)} uses the ML estimate of $\alpha$ given by $\hat{\alpha}_{\rm ML} = \frac{1}{T_{\rm est}} \sum_{t=1}^{T_{\rm est}}\widetilde{y}_t$. 
For $T_{\rm total}=16$ budget, \emph{RM(known $\alpha$)} outperforms \emph{CP(subspace)}, confirming the benefit of exact channel knowledge. When $\alpha$ is unknown, however, \emph{CP(subspace)} outperforms \emph{RM(unknown $\alpha$)} for an equal total measurement budget. 
Although increasing $T_{\rm est}$ moves coherent RM decoding toward its known-$\alpha$ benchmark, using these measurements for subspace coding improves the DoA estimation performance more. 
\Cref{fig:t_est_sweep_with_cp}(b) 
reaches the same conclusion  for a uniformly random on-grid DoA and $\alpha\sim\mathcal{CN}(0,1)$, where \emph{RM(unknown $\alpha$)} uses the MMSE estimate of $\alpha$ given by $\hat{\alpha}_{\rm MMSE} = \frac{\sigma_{s}^2}{T_{\rm est}\sigma_{s}^2 + \sigma_n^2}\sum_{t=1}^{T_{\rm est}} \widetilde{y}_t$. 
The plots decay slower due to deep fades, but CP(subspace) again gives the lower error probability at fixed total measurement budget. 

Hence, the subspace coding approach 
outperforms coherent decoding that estimates $\alpha$ using different sample budgets. This aligns with the fact 
the employed BGCs have guaranteed large subspace distance \cite{Mahdavifar2022} (important for noncoherent decoding), 
whereas BPSK-mapped Reed–Muller codes are not guaranteed to be optimal in Euclidean distance (important for coherent decoding). We also note that when good BGCs are unavailable, the performance gap may narrow, and coherent strategies that estimate $\alpha$ may become more competitive if they are optimized to ensure favorable Euclidean distance geometry.


\section{Conclusions}\label{sec:conclusions}
This work developed beamspace Grassmannian codes (BGCs) for nonadaptive mmWave channel sensing with unknown complex path gain, where DoA estimation naturally becomes a subspace decoding problem. 
We showed that sensing error is jointly governed by beamforming gain and the minimum subspace distance induced by the beamformer-array pair. We identified two array-dependent regimes, determined by the rank of the array-manifold matrix: one in which arbitrary Grassmannian line packings are realizable, and another in which realizability is constrained by compatibility with the array manifold. We further studied two settings depending on available information about the DoA. Without prior DoA information, we constructed spatially isotropic BGCs with large minimum subspace distance and characterized the subspace distance of BPSK-mapped binary linear codes, clarifying the role of Hamming distance in sensing. With prior angular information, sparse-array and convolutional-beamspace constructions jointly provide directional gain and favorable subspace distance. Our findings also motivate several new directions, including array-manifold-aware beamformer design in partially realizable regimes, pruning of sensing codebooks, and modulation design. Progress along these directions can enable resource-efficient sensing with rigorous guarantees in future large-scale MIMO systems.

\bibliographystyle{ieeetr}
\bibliography{references,references_external,references_v2}


\appendices\crefalias{section}{appendix}
\renewcommand{\thesectiondis}[2]{\Alph{section}:}
\clearpage
\pagenumbering{arabic}


\section{Proof of \Cref{thm:ML_decoder_error_exp}}\label{apdx:error_exponent_proof}
{
Let \(\b_n=\W\a_{\SS}(f_n)\) and \(\widehat{\b}_n=\b_n/\|\b_n\|_2\). Since \(f^\star=f_k\), the received vector can be written as
    \begin{align}
        \y
        =
        \alpha \b_k+\z
        =
        \alpha \|\b_k\|_2 \widehat{\b}_k+\z,    
    \end{align}
    where $\z\sim\mathcal{CN}(0,\sigma^2\I_T)$ due to design choice \labelcref{d:unitnorm}.
    The MD decoder makes an error if for 
    some 
    \(f_\ell\), \(\ell\neq k\), we have $ |\y^H\widehat{\b}_\ell|^2 > |\y^H\widehat{\b}_k|^2$. Hence, the probability of error can be upper bounded using the union bound as
    \begin{align}
        &\PP(\widehat f_{\rm md}(\y)\neq f_k \mid f^\star=f_k) \nonumber \\
        &=
        \PP\!\left(
        \bigcup_{\ell\neq k}
        \left\{
        |\y^H\widehat{\b}_\ell|^2
        >
        |\y^H\widehat{\b}_k|^2
        \right\}
        \,\middle|\, f^\star=f_k
        \right) \nonumber \\
        &\leq
        \sum_{\ell\neq k}
        \PP\!\left(
        |\y^H\widehat{\b}_\ell|^2
        >
        |\y^H\widehat{\b}_k|^2
        \,\middle|\, f^\star=f_k
        \right). \label{eq:Pe_union_bound}
    \end{align}
    It remains to evaluate the pairwise term. For each pair \(k\neq \ell\),
    the event $\{|\y^H\widehat{\b}_\ell|^2 > |\y^H\widehat{\b}_k|^2\}$
    is the standard single-source deterministic ML pairwise error event where $f^\star$ and $\alpha$ are deterministic unknowns \cite{Athley2005}\cite[Sec 9.3.1]{athley2003space}. Denote the subspace distance between $\Span{\widehat{\b}_{\ell}}$ and $\Span{\widehat{\b}_{k}}$ by 
    \[d_{k,\ell}^{(s)}\triangleq 1-|\widehat{\b}_k^H\widehat{\b}_\ell|^2.\]
    Therefore, applying the closed-form expression for pairwise error probability (PEP) of single-source deterministic unknown case with effective SNR $= \frac{|\alpha|^2\|\b_k\|^2_2}{\sigma^2}$ yields \cite[Eq. (9.38)]{athley2003space}
    \begin{align}
        P_{k\to l} &\triangleq \PP\!\left(
        |\y^H\widehat{\b}_\ell|^2
        >
        |\y^H\widehat{\b}_k|^2
        \,\middle|\, f^\star=f_k
        \right) \nonumber \\
        &= Q_1(u_{k,\ell},v_{k,\ell}) - \frac{1}{2}\exp\left(-\frac{|\alpha|^2\|\b_k\|_2^2}{2\sigma^2}\right) I_0(u_{k,\ell}v_{k,\ell}),
        \label{eq:exact_pairwise_pep}
    \end{align}
    where $Q_1(.,.)$ is the first-order Marcum Q-function, $I_0(.)$ is the zeroth-order modified Bessel function of the first kind, and 
    \begin{align*}
        u_{k,\ell}
        &=
        \sqrt{
        \frac{|\alpha|^2\|\b_k\|_2^2}{2\sigma^2}
        \left(1-\sqrt{d_{k,\ell}^{(s)}}\right)
        },\\
        v_{k,\ell}
        &=
        \sqrt{
        \frac{|\alpha|^2\|\b_k\|_2^2}{2\sigma^2}
        \left(1+\sqrt{d_{k,\ell}^{(s)}}\right)
        }.
    \end{align*}
    Since \(I_0(x) = \frac{1}{\pi}\int_{0}^{\pi}e^{x\cos \theta} d\theta >0\) for \(x\geq 0\), the second term in \eqref{eq:exact_pairwise_pep} can be removed to obtain the upper bound
    \begin{align*}
        P_{k\to \ell}
        \leq
        Q_1(u_{k,\ell},v_{k,\ell}).
    \end{align*}
    Since \(d_{k,\ell}^{(s)}\geq 0\), \(u_{k,\ell}/v_{k,\ell}\leq1\), we can use the exponential upper bound \(Q_1(a,b)\leq \exp(-(b-a)^2/2)\) \cite{corazza2002new} to obtain
    \begin{align}
        P_{k\to \ell}
        &\leq
        \exp\!\left( -\frac{(v_{k,\ell}-u_{k,\ell})^2}{2} \right) \nonumber  \\
        &= \exp\left(- \frac{|\alpha|^2\|\b_k\|_2^2}{2\sigma^2}
        \left(
        1-\sqrt{1-d_{k,\ell}^{(s)}}\right)\right).\label{eq:pep_exponential_bound_pairwisedist}
    \end{align}
    Finally, since \(d_{k,\ell}^{(s)}\geq d^{(s)}_{\min}(\cC)\) for all \(k\neq \ell\), and the function \(d\mapsto 1-\sqrt{1-d}\) is increasing on \([0,1]\), we obtain
    \begin{align}
        P_{k\to \ell}
        &\leq \exp\left(- \frac{|\alpha|^2\|\b_k\|_2^2}{2\sigma^2}
        \left( 1-\sqrt{1-d^{(s)}_{\min}(\cC)}\right)\right). \label{eq:pep_exponential_bound}
    \end{align}
    Substituting \eqref{eq:pep_exponential_bound} into the union bound \eqref{eq:Pe_union_bound} gives \eqref{eq:ml_ub}, 
    proving \cref{thm:ML_decoder_error_exp}.\hfill$\blacksquare$}

\section{Upper Bound on Minimum subspace distance of subsets of $\mathrm{RM}(2,4)$-induced BGC}

\label{apdx:rm_dmin_proof}

\begin{proposition}
\label{prop:rm24_subset_upper_bound}
Let $\cR=\mathrm{RM}(2,4)$ and let $\cC\subseteq\cR$ be a subset of size $|\cC|=M$. Define its BPSK-mapped codebook as $\widetilde{\cC}=\{\mathbf{1}-2\c:\c\in\cC\}$. Then
\begin{align}
d_{\min}^{(s)}(\widetilde{\cC}) \leq
\begin{cases}
1, & 2\leq M\leq16,\\
\frac{15}{16}, & 17\leq M\leq239,\\
\frac{3}{4}, & 240\leq M\leq1024,\\
0, & 1025\leq M\leq2048.
\end{cases}
\label{eq:rm24_subset_upper_bound}
\end{align}
\end{proposition}

\begin{proof}
The weight of a binary vector $\c$ is the number of non-zero entries in $\c$ denoted as $\mathrm{wt}(\c)\triangleq \|\c\|_0$. The weight enumerator of $\cR$ is a polynomial defined as 
\begin{align*}
    W_{\cR}(z) \triangleq \sum_{\c\in \cR} z^{\mathrm{wt}(\c)} = \sum_{n=0}^{16} |A_{\cR}(n)| z^n
\end{align*}
where $A_{\cR}(n) \triangleq \{\c\in \cR \text{ s.t. } \mathrm{wt}(\c)=n\} \subset\cR$. 
The weight enumerator of $\cR$ is given by \cite{sloane1970weight} (see \cite[eq. (49)]{hu2021mitigating}) 
\begin{align*}
W_{\cR}(z)\!=\!1+140z^4+448z^6+870z^8 +448z^{10}+140z^{12}+z^{16}.
\end{align*}
This implies the weight of $\c\in \cR$ only takes values in $\mathrm{wt}(\c)\in \{0,4,6,8,10,12,16\}$. Since $\cR$ is a linear codebook, the set of possible weights also corresponds to the set of possible Hamming distances between any two codewords in $\cR$. Hence, the Hamming distance between any two distinct codewords $\c_i,\c_j\in \cR$ takes value in
\begin{align}
d^{(h)}(\c_i,\c_j) \in \{4,6,8,10,12,16\}. \label{eq:dham_values_RM}
\end{align}
Let $\tilde{\c}_i=1-2\c_i,\tilde{\c}_j=1-2\c_j$ be the BPSK-mapped codewords. Using \cref{lem:bpsk_inner_prod}, the subspace distance between $\Span{\tilde{\c}_i}$ and $\Span{\tilde{\c}_j}$ is given by
\begin{align}
d^{(s)}(\tilde{\c}_i,\tilde{\c}_j) &= 1-\left( 1-\frac{2\dham(\c_i,\c_j)}{16} \right)^2 . \label{eq:ds_dham_relation}
\end{align}
Based on the values taken by $\dham(\c_i,\c_j)$ in \eqref{eq:dham_values_RM}, the subspace distance between $\Span{\tilde{\c}_i}$ and $\Span{\tilde{\c}_j}$ takes values 
\begin{align}
 \ds(\tilde{\c}_i,\tilde{\c}_j) \in \left\{0, \frac{3}{4},\frac{15}{16},1 \right\}   \label{eq:rm24_possible_subspace_dist}
\end{align}


Now, $d_{\min}^{(s)}(\widetilde{\cC})=1$ implies all codewords in $\widetilde{\cC}$ are mutually orthogonal. Since $\widetilde{\cC}\subset \mathbb{R}^{16}$, there can be at most $16$ mutually orthogonal vectors. 
Therefore
\begin{align}
d_{\min}^{(s)}(\widetilde{\cC}) \leq \frac{15}{16}, \qquad \forall M\geq17. \label{eq:ds_bound_M17}
\end{align}

We next show that no subset of size $M\geq240$ can have minimum subspace distance strictly larger than $3/4$. Assume there exists $\cC'\subset \cR$ such that $|\cC'|=M\geq 240$ and $d_{\min}^{(s)}(\widetilde{\cC'}) > \frac{3}{4}.$ By \cref{eq:rm24_possible_subspace_dist,eq:ds_bound_M17}, every pair $\widetilde{\c}'_i,\widetilde{\c}'_j\in\widetilde{\cC'}$, $i\neq j$, must satisfy
$\ds(\widetilde{\c}'_i,\widetilde{\c}'_j) = \frac{15}{16}$. Using \cref{eq:ds_dham_relation}, this implies 
\begin{align}
 \dham(\c'_i,\c'_j)\in \{6,10\}\ \forall \c'_i,\c'_j \in \cC',\ \c'_i\neq \c'_j. \label{eq:dham_bounds}
\end{align}

Define the complement closure of $\cC'$ as $\cD' \triangleq \cC'\cup(\cC'\oplus\mathbf{1}).$
Since $\cC'$ contains no complementary pair (otherwise $d_{\min}^{(s)}(\widetilde{\cC'}) = 0$; see \cref{rem:min_subspace_0}), the size of the complement closure is $|\cD'|=2M$ and $\cD' \subseteq \cR$ since $\cR$ is linear and $\mathbf{1}\in \cR$. 

Any element of $\u'\in\cD'$ can be written as
$\u'=\c'_i\oplus a\mathbf{1}$ for some $\c'_i\in\cC'$ and
$a\in\{0,1\}$. For any two binary vector $\x,\y\in \{0,1\}^{16}$
and $a,b\in\{0,1\}$, it can be verified that
\begin{align} 
\dham \left( \x\oplus a\mathbf{1}, \y\oplus b\mathbf{1} \right) =
\begin{cases}
\dham(\x,\y), & a=b,\\
16-\dham(\x,\y), & a\neq b.
\end{cases}
\label{eq:complement_hamming_distance}
\end{align}
Combining \eqref{eq:dham_bounds} with \eqref{eq:complement_hamming_distance} yields $ \dham(\u'_i,\u'_j)\!\in\!\{6,10,16\}\ \forall \u'_i,\u'_j \in \cD',\ \u'_i\neq \u'_j$ \footnote{Hamming distance $\dham(\u'_i,\u'_j) = 16$ for $\u'_i,\u'_j\in \cD'$ occurs when $\u'_i = \c'\in\cC'$ and $\u'_j = \c'\oplus \mathbf{1}$ due to complement closure.}, proving that $\dham_{\min}(\cD') \geq 6$. 

The binary Hamming bound states that for $t=\lfloor(\dham_{\min}(\cD')-1)/2\rfloor$, we have \cite[eq. (29)]{macwilliams1977theory}
\begin{align}
|\cD'| \sum_{r=0}^{t}\binom{16}{r} &\leq 2^{16}. \label{eq:binary_hamming_bound}
\end{align}
Now, $\dham_{\min}(\cD') \geq 6$ implies $t\geq 2$ which, 
per 
\eqref{eq:binary_hamming_bound} implies
\begin{align*}
|\cD'| \sum_{r=0}^{2}\binom{16}{r} = (2M)(137)\leq |\cD'| \sum_{r=0}^{t}\binom{16}{r}  &\leq 2^{16}.
\end{align*}
This gives the upper bound
\begin{align*}
M
\leq \left\lfloor \frac{65536}{2(137)} \right\rfloor = 239,
\end{align*}
which contradicts 
the 
assumption $M\geq240$. Therefore,
\begin{align*}
d_{\min}^{(s)}(\widetilde{\cC}) \leq \frac{3}{4}, \qquad \forall M\geq240.
\end{align*}

Finally, the $2048$ codewords of $\mathrm{RM}(2,4)$ form
$1024$ complementary pairs. Hence, if $M>1024$, the
pigeonhole principle guarantees that $\cC$ contains at least
one complementary pair which implies $\widetilde{\cC}$ contains atleast two codewords that span the same Grassmannian line, so
\begin{align*}
d_{\min}^{(s)}(\widetilde{\cC}) = 0, \qquad \forall M\geq 1025.
\end{align*}
This proves \cref{eq:rm24_subset_upper_bound}.
\end{proof}


\end{document}